\documentclass[12pt]{article}

\usepackage{mathtools}
\usepackage{lipsum}
\usepackage{mathtools}
\usepackage{multirow}
\usepackage{makecell}
\usepackage{caption}
\usepackage{graphicx}%
\usepackage{multirow}%
\usepackage{amsmath,amssymb,amsfonts}%
\usepackage{amsthm}%
\usepackage{mathrsfs}%
\usepackage[title]{appendix}%
\usepackage{xcolor}%
\usepackage{textcomp}%
\usepackage{manyfoot}%
\usepackage{booktabs}%
\usepackage{algorithm}%
\usepackage{algorithmicx}%
\usepackage{algpseudocode}%
\usepackage{listings}%
\usepackage{authblk}
\usepackage{adjustbox}
\usepackage{longtable} 
\usepackage{array}
\usepackage{tabularx}
\usepackage{manyfoot}
\usepackage{xcolor}
\usepackage{lscape}
\usepackage{multicol}%
\usepackage{textcomp}%
\usepackage{manyfoot}%
\usepackage{booktabs}%
\usepackage{array}
\usepackage{enumerate}
\usepackage{hyperref}
\usepackage{ mathrsfs }
\usepackage{autobreak}
\usepackage{mathtools}
\usepackage{upgreek}
\usepackage{multirow}
\usepackage{nicematrix}
\usepackage{tabularray}
\usepackage{adjustbox}
\usepackage{ tipa }
\usepackage{tabularx}
\usepackage{tabulary}
\usepackage{ragged2e}
\usepackage{array}
\usepackage{multicol}
\usepackage{nicematrix}
\usepackage{booktabs}
\usepackage{upgreek}
\usepackage{autobreak}
\usepackage{makecell}
\usepackage{float}
\usepackage{mathrsfs}
\usepackage{amsthm}
\newtheorem{lemma}{Lemma}

\theoremstyle{thmstyleone}%
\newtheorem{theorem}{Theorem}
\newtheorem{corollary}{Corollary}%
\theoremstyle{thmstyletwo}%
\newtheorem{example}{Example}%
\newtheorem{remark}{Remark}%

\theoremstyle{thmstylethree}%

\title{\textbf{Weight distributions of two classes of linear codes with few weights derived from Weil sums}}

\author[1]{Mrinal Kanti Bose}
\author[2]{Abhay Kumar Singh\thanks{Corresponding author.\\ $^1$e-mail: bose0261@gmail.com,\\ $^2$e-mail: abhay@iitism.ac.in}}
\affil[1,2]{Department of Mathematics and Computing, Indian Institute of Technology (ISM) Dhanbad, Dhanbad 826004, Jharkhand, India.}

\begin{document}
	
	\maketitle
	
	\begin{abstract}
		Linear codes with few weights have been a subject of study for many years, as they have applications in secret sharing, authentication codes, association schemes, and strongly regular graphs. In this article, two distinct classes of $p$-ary linear codes are constructed through the selection of two specific defining sets. Their weight distributions are completely determined for each case by detailed calculations on certain Weil sums. The constructed codes are shown to have only two, four, six, eight, and nine nonzero weights under different cases. In particular, we obtained an infinite family of two-weight optimal linear codes with respect to the Griesmer bound. Moreover, we observe that some of our newly constructed codes are minimal under certain conditions.
	\end{abstract}
	
\section{Introduction}\label{sec1}
 Let $q=p^{e}$, where $p$ is a prime and $e$ is a positive integer. Let $\mathbb{F}_{q}$ be a finite field of order $q$, and let $\mathbb{F}_{q}^{*}=\mathbb{F}_{q}\backslash\{0\}$ be the multiplicative group of order $q-1$. A linear code $\mathcal{C}$ with parameters $[n, k, d]$ over $\mathbb{F}_{q}$ is a $k$-dimensional subspace of $\mathbb{F}_{q}^{n}$, where $n$ denotes the length of the code $\mathcal{C}$, whereas $d$ represents the minimum (Hamming) distance of the code $\mathcal{C}$. An $[n,k,d]$ linear code $\mathcal{C}$ is said to be optimal if $n$, $k$, and $d$ meet a bound on linear codes \cite{P11}, and is called almost optimal if there exists an optimal $[n,k,d+1]$ linear code. The well-known Griesmer bound \cite{P9} for a linear $[n,k,d]$ code $\mathcal{C}$ over $\mathbb{F}_{q}$ with $k\geq 1$ states that
 \begin{equation*}
     n\geq \sum_{i=0}^{k-1}\Big\lceil\frac{d}{q^{i}}\Big\rceil,
 \end{equation*}
 where $\lceil\cdot \rceil$ denotes the ceiling function. For $1\leq i\leq n$, let $A_{i}=|\{c\in\mathcal{C}:\operatorname{wt}_{H}(c)=i\}|$, where $\operatorname{wt}_{H}(c)$ denotes the Hamming weight of the codeword $c$ in $\mathcal{C}$. Obviously, $A_{0}=1$. The weight enumerator of an $[n,k]$ linear code $\mathcal{C}$ is defined by the polynomial $1 +A_1z +A_2z^2 +\cdots+A_nz^n$. The sequence $(1, A_1, A_2,\cdots,A_n)$ is called the weight distribution of $\mathcal{C}$. We refer to $\mathcal{C}$ as a $t$-weight code if the cardinality of the set $\{1\leq i\leq n:A_{i}\neq 0\}$ is equal to $t$. The weight distribution is an important research topic in coding theory, not only because the weight enumerator of a code contains crucial information about its error correction capabilities, but also because it allows the computation of the error probability or error detection and correction with respect to some algorithms \cite{P21}. The reader is referred to \cite{P3,P22,P15,P10,R4} for known weight enumerators of some linear codes. Furthermore, linear codes with few weights have important applications in secret sharing schemes \cite{R0,R1, P22, P15}, association schemes \cite{P2}, strongly regular graphs \cite{P7}, and authentication codes \cite{P16}.
 \vskip 1 pt For two vectors $\bar{a}=(a_0,a_1,\cdots,a_{n-1})$ and $\bar{b}=(b_0,b_1,\cdots,b_{n-1})$ in $\mathbb{F}_{q}^{n}$, we say that $\bar{a}$ covers $\bar{b}$, if $\operatorname{Supp}(\bar{b})\subseteq\operatorname{Supp}(\bar{a})$, where $\operatorname{Supp}(\bar{a})=\{1\leq i\leq n: u_{i}\neq 0\}$ is the support of $\bar{a}$. A codeword $\bar{a}$ in $\mathcal{C}$ over $\mathbb{F}_{q}$ is said to be minimal if $\bar{a}$ covers only the codeword $\lambda\bar{a}$ for all $\lambda\in\mathbb{F}_{q}$, but no other codewords in $\mathcal{C}$. A linear code $\mathcal{C}$ over $\mathbb{F}_{q}$ is said to be minimal if every codeword in $\mathcal{C}$ is minimal. The construction of minimal linear codes is meaningful work because of their important applications in secret sharing schemes \cite{R0,P22,P15}. The following lemma introduces a sufficient condition, which is often used to determine whether a linear code over $\mathbb{F}_{q}$ is minimal, is presented as follows.
 \vspace{1em}
 \begin{lemma}$\textnormal{\cite{P1}}$
     A linear code $\mathcal{C}$ over $\mathbb{F}_{q}$ is minimal if $\frac{w_{\textit{min}}}{w_{\textit{max}}}>\frac{q-1}{q}$, where $w_{\textit{min}}$ and $w_{\textit{max}}$ denote the minimum and maximum nonzero Hamming weights in $\mathcal{C}$, respectively. 
 \end{lemma}
 \vspace{1em}
 \vskip 1 pt  Ding and Niederreiter \cite{P6} introduced a fundamental method for constructing linear codes based on the proper selection of a subset of finite fields. For $D\subset \mathbb{F}_{p^e}$, they defined 
\begin{equation}\label{EQ1}
    \mathcal{C}_{D}=\{c(\gamma)=(\operatorname{Tr}_{1}^{e}(\gamma x))_{x\in D}:\gamma\in\mathbb{F}_{p^e}\},
\end{equation}
where $\operatorname{Tr}_{1}^{e}(x)=\sum_{i=0}^{e-1}x^{p^{i}}$ is the trace function from $\mathbb{F}_{p^e}$ to $\mathbb{F}_{p}$. From the definition of $(\ref{EQ1})$, it is clear that $\mathcal{C}_{D}$ forms a linear code of length $|D|$ over $\mathbb{F}_{p}$ and the set $D$ is called the defining set of $\mathcal{C}_{D}$. In fact, several classes of good linear codes with few weights have been constructed using the defining-set approach over the past few decades (see \cite{P5, S4, P4, S5, P3, P12, P15, P18, P22, P13, P23, R2, R3, P24}).
\vskip 1pt Recently, Cheng and Gao \cite{P8} proposed a new method for evaluating a binomial Weil sum given by
\begin{equation}\label{EQ0}
    S_{\mathcal{N}}(a,b)=\sum_{x\in\mathbb{F}_{Q}^{*}}\chi_{\phi(\mathcal{N})}\left(ax^{\frac{Q-1}{\mathcal{N}}}+bx\right),
\end{equation}
where $a,b\in\mathbb{F}_{Q}$, $Q=p^{\phi(\mathcal{N})}$, $\mathcal{N}$ is a positive integer not divisible by an odd prime $p$, and $\operatorname{ord}_{\mathcal{N}}(p)=\phi(\mathcal{N})$. Specifically, they provide the evaluation of $S_{\mathcal{N}}(a,b)$ for $\mathcal{N}=2,4,\ell^{k}\text{ or }2\ell^{k}$, where $\ell$ is an odd prime satisfying $p\nmid\ell$ and $k$ is a positive integer. Taking $\mathcal{N}=\ell^k$, they then construct a new class of two-weight ternary linear codes $\mathcal{C}_{D}$ of the form $(\ref{EQ1})$ by choosing the defining set $D=\{x\in\mathbb{F}_{Q}^{*}:\operatorname{Tr}_{1}^{\phi(\mathcal{N})}(x^{(Q-1)/\mathcal{N}})=0\}$. 
Subsequently, Cheng \cite{R3} further investigated the ternary linear code $\mathcal{C}_{D}$ of the form $(\ref{EQ1})$ to construct some classes of two-weight and four-weight linear codes, selecting the defining set as follows
\begin{equation}\label{EQ23}
    D=\{x\in\mathbb{F}_{Q}^{*}:\operatorname{Tr}_{1}^{\phi(\mathcal{N})}(x^{\frac{Q-1}{\mathcal{N}}}+\beta x)=\alpha\},
\end{equation}
where $\mathcal{N}=2\ell^{k}$, $Q=3^{\phi(\mathcal{N})}$, $\alpha\in\mathbb{F}_{3}$, $\beta\in\mathbb{F}_{Q}$, and $\operatorname{ord}_{\mathcal{N}}(3)=\phi(\mathcal{N})$. 
Very recently, Cheng and Sheng \cite{P20} extended the defining set $D$ as in $(\ref{EQ23})$ to an arbitrary odd prime $p$ to investigate the linear code $\mathcal{C}_{D}$. 
\vskip 1pt \par
Li et al. \cite{P5} considered the linear code of the form 
\begin{equation}\label{EQ2}
    \mathcal{C}_{D}=\{c(\gamma,\delta)=(\operatorname{Tr}_{1}^{e}(\gamma x+\delta y))_{(x,y)\in D}:\gamma,\delta\in\mathbb{F}_{q}\},
\end{equation}
where $D\subseteq\mathbb{F}_{q}^{2}$ is also called a defining set. They investigated the linear code $\mathcal{C}_{D}$ of the form $(\ref{EQ2})$ for $D=\{(x, y)\in \mathbb{F}_{q}^{2}\backslash\{(0,0)\}:\operatorname{Tr}_{1}^{m}(x^{N_1}+y^{N_2})=0\}$, where $N_1, N_2 \in\{1,2, p^{\frac{e}{2}}+1\}$. 
Motivated by the work of \cite{P5}, Wu et al. \cite{P12} constructed two classes of two-weight and three-weight linear codes $\mathcal{C}_{D}$ of the form $(\ref{EQ2})$ and presented their weight distributions, with the defining set $D=\{(x,y)\in\mathbb{F}_{q}^{2}\backslash\{(0,0)\}:f(x)+g(y)=0\}$ by considering the two cases $(1)$ $f(x)=\operatorname{Tr}_{1}^{e}(x)$ and $g(y)$ is a weakly regular bent function; and $(2)$ $f(x)$ and $g(y)$ are both weakly regular bent functions.
\vskip 1 pt
Throughout this paper, we assume $q=p^{e}$, where $e=\phi(\ell^{k})$, $p$ and $\ell$ are distinct odd primes, $k$ is a positive integer, and $p$ is a primitive root modulo $2\ell^{k}$. That means $\operatorname{ord}_{2\ell^{k}}(p)=\phi(\ell^{k})$. Motivated by the prior works, in this paper, we investigate the linear code $\mathcal{C}_{D}$ of the form $(\ref{EQ2})$ by using the following two defining sets:
 \begin{equation}\label{EQ9}
     D_{u}=\{(x,y)\in\mathbb{F}_{q}^{2}\backslash\{(0,0)\}:\operatorname{Tr}_{1}^{e}(x+y^{N})=u\},
 \end{equation}
 and 
 \begin{equation}\label{EQ10}
     D^{'}=\{(x,y)\in\mathbb{F}_{q}^{2}\backslash\{(0,0)\}:f(x)+\operatorname{Tr}_{1}^{e}(y^{N})=0\},
 \end{equation}
 where $p$ is any odd prime, $q=p^{e}$, $e=\phi(\ell^{k})=(\ell-1)\ell^{k-1}$, $u\in\mathbb{F}_{p}$, $N=\frac{q-1}{2\ell^{k}}$, and in $(\ref{EQ10})$, we consider $f:\mathbb{F}_{q}\rightarrow\mathbb{F}_{p}$ to be a weakly regular bent function such that $f^{*}$ satisfies the condition $f^{*}(cx)=c^{2}f^{*}(x)$ for any $c\in\mathbb{F}_{p}^{*}$ and $x\in\mathbb{F}_{q}$, where $f^{*}$ is the dual of $f$ discussed in section \ref{S1}. Through detailed calculations on certain exponential sums, we obtain five classes of $4$-weight linear codes, one class of $2$-weight linear codes, two classes of $6$-weight linear codes, two classes of $8$-weight linear codes, and one class of $9$-weight linear codes from our construction. It is noteworthy that the defining sets in $(\ref{EQ9})$ and $(\ref{EQ10})$ used to develop linear codes in this paper deal with a new binomial Weil sum $S_{\mathcal{N}}(a,b)$ defined in $(\ref{EQ0})$ comparing the previous works summarized in Table \ref{Table7}. The parameters and weight distributions of these codes are completely characterized by using Weil sums and quadratic Gauss sums. In addition, we identify one infinite class of two-weight optimal codes that meets the Griesmer bound. 
\vskip 1 pt
The rest of this paper is organized as follows. Section $\ref{sec2}$ introduces some basic knowledge on exponential sums, cyclotomic fields, weakly regular bent functions, and related auxiliary results that will be used in the subsequent sections. In section $\ref{Sec3}$ and section $\ref{sec4}$, we first present some auxiliary results, and then we respectively give the constructions of $p$-ary linear codes from the defining sets $(\ref{EQ9})$ and $(\ref{EQ10})$, with their parameters and weight distributions explicitly determined. Section $\ref{sec5}$ concludes this paper. 


    \begin{table}[htbp]
    \centering
    \footnotesize 
    \renewcommand{\arraystretch}{1.3} 
    \setlength{\tabcolsep}{9pt} 
    \begin{adjustbox}{max width=0.65\textwidth, center} 
        \begin{tabular}{|>{\centering\arraybackslash}m{6cm}|>{\centering\arraybackslash}m{4.6cm}|>{\centering\arraybackslash}m{7cm}|>{\centering\arraybackslash}m{3cm}|>{\centering\arraybackslash}m{2cm}|}
            \hline
            \textbf{\normalsize The defining set $D$ } & \textbf{\normalsize Conditions on $D$} & \textbf{\normalsize Parameters of $\mathcal{C}_{D}$ } & \textbf{\normalsize  Nonzero $t$-weights in $\mathcal{C}_{D}$} & \textbf{\normalsize References} \\
            \hline
           \multirow{7}{*}{\parbox{5.5cm}{$\{(x,y)\in\mathbb{F}_{p^{m}}^{2}\backslash\{(0,0)\}:\operatorname{Tr}_{1}^{m}(x^{N_{1}}+y^{N_{2}})=0\}$}}  & $N_{1}=N_{2}=1$  & $[p^{2m-1}-1,2m-1,p^{2m-1}-p^{2m-2}]$ &  $t=1$ & \cite{P5} \\
              & $N_{1}=1$, $N_{2}=2$ and $m\geq 2$ is even & $[p^{2m-1}-1,2m]$ & $t=3$ & \cite[Table 1]{P5} \\
              & $N_{1}=1$, $N_{2}=2$ and $m\geq 1$ is odd  & $[p^{2m-1}-1,2m]$ & $t=3$ &  \cite[Table 2]{P5} \\
              & $N_{1}=1$ and $N_{2}=p^{m/2}+1$, $m$ is an even integer & $[p^{2m-1}-1,2m,p^{2m-1}-p^{2m-2}-(p-1)p^{\frac{3m-4}{2}}]$ & $t=3$ &  \cite[Table 3]{P5} \\
              & $N_{1}=N_{2}=2$ & $[p^{2m-1}+(-1)^{\frac{m(p-1)}{2}}(p-1)p^{m-1}-1,2m]$ &  $t=2$ &  \cite[Table 4]{P5} \\
              & $N_{1}=2$ and $N_{2}=p^{m/2}+1$, $m$ is an even integer & $[p^{2m-1}+(-1)^{\frac{m(p-1)}{4}}(p-1)p^{m-1}-1,2m]$ & $t=2$ & \cite[Table 5]{P5} \\
              & $N_{1}=N_{2}=p^{m/2}+1$, $m$ is an even integer & $[p^{2m-1}+(p-1)p^{m-1}-1,2m]$ & $t=2$ & \cite[Table 6]{P5} \\
            \hline
            \multirow{4}{*}{\parbox{5.5cm}{$\{(x,y)\in\mathbb{F}_{p^{m}}^{2}\backslash\{(0,0)\}:\operatorname{Tr}_{1}^{m}(x+y^{p^{u}+1})=0\}$}} & $m$ is odd & $[p^{2m-1}-1,2m]$ & $t=3$ & \cite[Table 1]{P4} \\
            & $\frac{m}{\operatorname{gcd}(m,u)}$ is odd and $\operatorname{gcd}(m,u)$ is even   & $[p^{2m-1}-1,2m]$  & $t=3$ &  \cite[Table 2]{P4}  \\
            & $\frac{m}{\operatorname{gcd}(m,u)}\equiv 2\pmod{4}$ & $[p^{2m-1}-1,2m]$  & $t=3$ & \cite[Table 3]{P4} \\
            &  $\frac{m}{\operatorname{gcd}(m,u)}\equiv 0\pmod{4}$ & $[p^{2m-1}-1,2m]$  & $t=3$ & \cite[Table 4]{P4} \\
            \hline
            \multirow{2}{*}{\parbox{5.5cm}{$\{(x,y)\in\mathbb{F}_{p^{m}}^{2}\backslash\{(0,0)\}:\operatorname{Tr}_{1}^{m}(x^{2}+y^{p^{u}+1})=0\}$}} & $\frac{m}{\operatorname{gcd}(m,u)}$ is odd or $\frac{m}{\operatorname{gcd}(m,u)}\equiv 2 \pmod{4}$ & $[n,2m]$, where $n=
                (p^{m}+1)(p^{m-1}-1),\text{ if }p\equiv 3\pmod{4}\text{ and }\operatorname{gcd}(m,u)\text{ is odd;}$ and $n=(p^{m}-1)(p^{m-1}+1),\text{ otherwise}$
             & $t=2$ &  \cite[Table 5]{P4} \\
            & $\frac{m}{\operatorname{gcd}(m,u)}\equiv 0 \pmod{4}$ & $[p^{2m-1}+p^{m+v}-p^{m+v-1}-1,2m]$, where $v=\operatorname{gcd}(m,u)$ & $t=3$ & \cite[Table 6]{P4} \\
            \hline
            \multirow{3}{*}{\parbox{5.5cm}{$\{(x,y)\in\mathbb{F}_{p^{m}}^{2}:x\in C_{i},y\in C_{j}\}$, where $C_i$ and $C_j$ are any two cyclotomic classes of order $e$, $p^{m}\equiv 1\pmod{e}$ and $0\leq i,j\leq e-1$.}} & $m=2d\gamma$ and $d$ is the least positive integer satisfying $p^{d}\equiv -1\pmod{e}$  & $[\frac{(p^{m}-1)^{2}}{e^{2}},2m]$ & $t=5$ & \cite[Table 1]{P12} \\
            & $e=2$ and $m$ is even & $[\frac{(p^{m}-1)^{2}}{4},2m]$ & $t=5$ & \cite[Table 2]{P12} \\
            & $e=2$ and $m$ is odd  & $[\frac{(p^{m}-1)^{2}}{4},2m]$ & $t=2$   & \cite[Table 3]{P12}   \\           
           \hline
           \multirow{4}{*}{\parbox{5.5cm}{$\{(x,y)\in\mathbb{F}_{p^{m}}^{2}\backslash\{(0,0)\}:f(x)+g(y)=0\}$, where $f$ and $g$ are weakly regular bent functions from $\mathbb{F}_{p^{m}}$ to $\mathbb{F}_{p}$.}} & $f(x)=\operatorname{Tr}_{1}^{m}(x)$ and $m$ is even  & $[p^{2m-1}-1,2m]$ & $t=3$ & \cite[Table 5]{P12} \\
            & $f(x)=\operatorname{Tr}_{1}^{m}(x)$ and $m$ is odd & $[p^{2m-1}-1,2m]$ & $t=3$ & \cite[Table 6]{P12} \\
            & $f$ and $g$ are weakly regular bent functions in \cite[Table 4]{P12}, $l_{f}, l_{g}\in\{2,p-1\}$, $l_{f}\neq l_{g}$ and $m$ is even  & $[p^{2m-1}+\frac{p-1}{p}\epsilon_{f}\epsilon_{g}p^{*m}-1,2m]$ & $t=3$   &  \cite[Table 7]{P12}   \\   
            & $f$ and $g$ are weakly regular bent functions in \cite[Table 4]{P12}, $l_{f}, l_{g}\in\{2,p-1\}$, $l_{f}=l_{g}$ and $m$ is even  & $[p^{2m-1}+\frac{p-1}{p}\epsilon_{f}\epsilon_{g}p^{*m}-1,2m]$ & $t=2$   & \cite[Table 8]{P12}   \\
           \hline
           \multirow{2}{*}{\parbox{6cm}{$\{(x,y)\in\mathbb{F}_{p^{2m}}^{2}\backslash\{(0,0)\}:\operatorname{Tr}_{1}^{2m}(x^{p^{l}+1})=1,\operatorname{Tr}_{1}^{2m}(y)\in C_{i}^{(2,p)}\}$, where $l$ and $m$ are positive integers, $C_{i}^{(2,p)}$ are the sets of all squares and non-squares in $\mathbb{F}_{p}^{*}$ $\text{resp. for }i=0\text{ and }1$.}} & $p\equiv 3\pmod{4}$, $\frac{2m}{\operatorname{gcd}(l,2m)}$ is even and $\frac{m}{\operatorname{gcd}(l,2m)}\equiv 1\pmod{2}$ is even  & $[\frac{p-1}{2}(p^{4m-2}+p^{3m-2}),4m]$ & $t=6$ & \cite[Table 1]{R4} \\ [15pt]
            & $p\equiv 3\pmod{4}$, $\frac{2m}{\operatorname{gcd}(l,2m)}$ is even, $m\geq \operatorname{gcd}(l,2m)+1$ and $\frac{m}{\operatorname{gcd}(l,2m)}\equiv 0\pmod{2}$ is even & $[\frac{p-1}{2}(p^{4m-2}+p^{3m+\operatorname{gcd}(l,2m)-2}),4m]$ & $t=6$ & \cite[Table 2]{R4} \\ [15pt]
           \hline 
            \multirow{2}{*}{\parbox{6cm}{$\{(x,y)\in\mathbb{F}_{p^{m}}^{2}\backslash\{(0,0)\}:\operatorname{Tr}_{1}^{m}(x)+g(y)=0\}$, where $g$ is a weakly regular $s$-plateaued unbalanced function from $\mathbb{F}_{p^{m}}$ to $\mathbb{F}_{p}$.}} & $m+s$ is even  & $[p^{2m-1}-1,2m]$ & $t=3$ & \cite[Table 1]{R5} \\ [18pt]
            & $m+s$ is odd & $[p^{2m-1}-1,2m]$ & $t=3$ & \cite[Table 2]{R5} \\
           \hline 
           \multirow{3}{*}{\parbox{6cm}{$\{(x,y)\in\mathbb{F}_{p^{m}}^{2}\backslash\{(0,0)\}:\operatorname{Tr}_{1}^{m}(x)+g(y)\in SQ(\text{}NSQ)\}$, where $g$ is a weakly regular $s$-plateaued unbalanced function from $\mathbb{F}_{p^{m}}$ to $\mathbb{F}_{p}$ and $SQ(\text{}NSQ)$ is the set of all squares (non-squares) in $\mathbb{F}_{p}^{*}$.}}  & $m+s$ is even  & $[\frac{p-1}{2}p^{2m-1},2m]$ & $t=3$ &  \cite[Table 3]{R5} \\
            & $m+s$ is odd and $p\equiv 1\pmod{4}$ & $[\frac{p-1}{2}p^{2m-1},2m]$ & $t=3$ &  \cite[Table 4 (Table 6)]{R5} \\
            & $m+s$ is odd and $p\equiv 3\pmod{4}$ & $[\frac{p-1}{2}p^{2m-1},2m]$ & $t=3$ &  \cite[Table 5 (Table 7)]{R5} \\
           \hline 
           \multirow{5}{*}{\parbox{5cm}{$\{(x,y)\in\mathbb{F}_{p^{m}}^{2}\backslash\{(0,0)\}:f(x)+g(y)\in T\}$, where $f$ and $g$ are weakly regular $s$-plateaued unbalanced functions from $\mathbb{F}_{p^{m}}$ to $\mathbb{F}_{p}$, $T\subset\mathbb{F}_{p}$ and $SQ(\text{}NSQ)$ denote the set of all squares (non-squares) in $\mathbb{F}_{p}^{*}$.}}  & $T=\{0\}$, $l_{f}=l_{g}$  & $[p^{2m-1}+\frac{p-1}{p}\epsilon_{f}\epsilon_{g}p^{*(m+s)}-1,2m]$ & $t=3$ &  \cite[Table 8]{R5} \\
            & $T=\{0\}$, $l_{f}\neq l_{g}$ and $m-s$ is even & $[p^{2m-1}+\frac{p-1}{p}\epsilon_{f}\epsilon_{g}p^{*(m+s)}-1,2m]$ & $t=4$ &  \cite[Table 9]{R5} \\
            & $T=\{0\}$, $l_{f}\neq l_{g}$ and $m-s$ is odd and $p\equiv 1\pmod{4}$ & $[p^{2m-1}+\frac{p-1}{p}\epsilon_{f}\epsilon_{g}p^{*(m+s)}-1,2m]$ & $t=4$ & \cite[Table 10]{R5} \\
            & $T=\{0\}$, $l_{f}\neq l_{g}$ and $m-s$ is odd and $p\equiv 3\pmod{4}$ & $[p^{2m-1}+\frac{p-1}{p}\epsilon_{f}\epsilon_{g}p^{*(m+s)}-1,2m]$ & $t=4$ & \cite[Table 11]{R5} \\
            & $T=SQ(\text{or, }NSQ)$, $l_{f}= l_{g}$ & $[p^{2m-1}\frac{p-1}{2}-\frac{p(p-1)}{2}\epsilon_{f}\epsilon_{g}p^{*(m+s-2)},2m]$ & $t=3$ & \cite[Table 13]{R5} \\
           \hline
           \multirow{2}{*}{\parbox{5.8cm}{$\{(x,y)\in\mathbb{F}_{p^{m}}^{2}\backslash\{(0,0)\}:\operatorname{Tr}_{1}^{m}(x)+g(y)=0\}$, where $g$ is a weakly regular $s_{g}$-plateaued balanced function from $\mathbb{F}_{p^m}$ to $\mathbb{F}_{p}$ with $1\leq s_{g}<m$.}} & $m+s_{g}$ is even & $[p^{2m-1}-1,2m]$ & $t=3$ & \cite[Table 1]{R7} \\ [18pt]
           & $m+s_{g}$ is odd & $[p^{2m-1}-1,2m]$ & $t=3$ &  \cite[Table 2]{R7} \\
           \hline
           \multirow{3}{*}{\parbox{5.5cm}{$\{(x,y)\in\mathbb{F}_{p^{m}}^{2}\backslash\{(0,0)\}:f(x)+g(y)=0\}$, where $f$ and $g$ are weakly regular $s_{f}$ and $s_{g}$-plateaued unbalanced ($\mathcal{WRP}$) function from $\mathbb{F}_{p^m}$ to $\mathbb{F}_{p}$ with $1\leq s_{f}, s_{g}<m$.}} & $l_{f}=l_{g}$, $m+s_{f}$ is odd and $m+s_{g}$ is even & $[p^{2m-1}-1,2m,(p-1)p^{2m-2}-\sqrt{p}^{2m+s_{f}+s_{g}-3}]$ & $t=3$ &  \cite[Table 3]{R7} \\
           & $2m+s_{f}+s_{g}$ is even, $0\leq s_{f},s_{g}< m-1$ and $l_{f}=l_{g}$ & $[p^{2m-1}-1+\epsilon_{f}\epsilon_{g}\frac{p-1}{p}\sqrt{p^{*}}^{2m+s_{f}+s_{g}},2m]$ & $t=3$ &  \cite[Table 4]{R7} \\
            & $2m+s_{f}+s_{g}$ is even, $0\leq s_{f},s_{g}< m-1$, $p>3$ and $l_{f}\neq l_{g}$  & $[p^{2m-1}-1+\epsilon_{f}\epsilon_{g}\frac{p-1}{p}\sqrt{p^{*}}^{2m+s_{f}+s_{g}},2m]$ & $t=4$ & \cite[Table 5]{R7} \\
            \hline
             \multirow{2}{*}{\parbox{5.5cm}{$\{(x,y)\in\mathbb{F}_{p^{m}}^{2}\backslash\{(0,0)\}:f(x)+g(y)=0\}$, where $f$ and $g$ are weakly regular $s_{f}$ and $s_{g}$-plateaued balanced ($\mathcal{WRPB}$) function from $\mathbb{F}_{p^m}$ to $\mathbb{F}_{p}$ with $1\leq s_{f}, s_{g}<m-1$.}} & $l_{f}=l_{g}$ and $2m+s_{f}+s_{g}$ is even & $[p^{2m-1}-1,2m]$ & $t=3$ &  \cite[Table 6]{R7} \\ [18pt]
           & $l_{f}\neq l_{g}$ and $2m+s_{f}+s_{g}$ is even & $[p^{2m-1}-1,2m]$ & $t=4$ &  \cite[Table 7]{R7} \\
           \hline
           \multirow{2}{*}{\parbox{5.5cm}{$\{(x,y)\in\mathbb{F}_{p^{m}}^{2}\backslash\{(0,0)\}:\operatorname{Tr}_{1}^{m}(x^{p}+xy)=u\}$, where $u\in\mathbb{F}_{p}$.}} & $u=0$ & $[p^{2m-1}+(p-1)p^{m-1}-1,2m, (p-1)p^{2m-2}]$ & $t=3$ &  \cite[Table 1]{R6} \\
           & $u\neq 0$ & $[p^{2m-1}-p^{m-1},2m, (p-1)p^{2m-2}-2p^{m-1}]$ & $t=3$ &  \cite[Table 2]{R6} \\
           \hline
            \multirow{6}{*}{\parbox{5.5cm}{$\{(x,y)\in\mathbb{F}_{p^{m}}^{2}\backslash\{(0,0)\}:\operatorname{Tr}_{1}^{m}(x+y^{N})=u\}$, where $m=\phi(\ell^{k})$, $k\in\mathbb{N}$, $u\in\mathbb{F}_{p}$, $\ell(\neq p)$ is an odd integer such that $\operatorname{ord}_{2\ell^{k}}(p)=\phi(\ell^{k})$ and $N=\frac{p^{m}-1}{2\ell^{k}}$.}} & $u=0$ and $\ell\equiv 1\pmod{p}$ & $[p^{2m-1}-1,2m]$ & $t=4$ & Table \ref{Table8} \\
            & $u=0$ and $\ell\not\equiv 1\pmod{p}$ & $[p^{2m-1}-1,2m]$ & $t=4$ & Table \ref{Table9} \\
             & $u\neq 0$ and $u^{2}\not\equiv\phi(\ell^{k})^{2}\pmod{p}$ and $u^{2}\not\equiv\ell^{2(k-1)}\pmod{p}$ & $[p^{2m-1},2m, p^{2m-2}(p-1)]$ & $t=2$ & Table \ref{Table10} \\
              & $u\neq 0$ and $u^{2}\equiv\phi(\ell^{k})^{2}\pmod{p}$ and $u^{2}\not\equiv\ell^{2(k-1)}\pmod{p}$ & $[p^{2m-1},2m]$ & $t=4$ & Table \ref{Table11} \\
               & $u\neq 0$ with $u\equiv\phi(\ell^{k})\equiv\ell^{k-1}\pmod{p}$ or $u\equiv-\phi(\ell^{k})\equiv-\ell^{k-1}\pmod{p}$ & $[p^{2m-1},2m]$ & $t=4$ & Table \ref{Table12} \\
                & $u\neq 0$ and $u^{2}\not\equiv\phi(\ell^{k})^{2}\pmod{p}$ and $u^{2}\equiv\ell^{2(k-1)}\pmod{p}$ & $[p^{2m-1},2m]$ & $t=4$ & Table \ref{Table13} \\
            \hline   
            \multirow{5}{*}{\parbox{5.5cm}{$\{(x,y)\in\mathbb{F}_{p^{m}}^{2}\backslash\{(0,0)\}:\operatorname{Tr}_{1}^{m}(f(x)+y^{N})=0\}$, where $m=\phi(\ell^{k})$, $k\in\mathbb{N}$, $u\in\mathbb{F}_{p}$, $\ell(\neq p)$ is an odd integer such that $\operatorname{ord}_{2\ell^{k}}(p)=\phi(\ell^{k})$, $N=\frac{p^{m}-1}{2\ell^{k}}$ and $f$ is a weakly regular bent function from $\mathbb{F}_{p^m}$ to $\mathbb{F}_{p}$ such that $f^{*}$ is a quadratic form.}} & $p\equiv 1\pmod{4}$ and $\ell\equiv 1\pmod{p}$ & $[p^{2m-1}+\epsilon_{f}\sqrt{p^{*}}^{m}\left(p^{m-1}(p-1)-(\ell-1)\frac{p^{m}-1}{\ell^{k}}\right)-1,2m]$ & $t=8$ & Table \ref{Table1} \\
            & $p\equiv 3\pmod{4}$ and $\ell\equiv 1\pmod{p}$ & $[p^{2m-1}+\epsilon_{f}\sqrt{p^{*}}^{m}\left(p^{m-1}(p-1)-(\ell-1)\frac{p^{m}-1}{\ell^{k}}\right)-1,2m]$ & $t=6$ & Table \ref{Table2} \\
             & $p\equiv 1\pmod{4}$, $\ell\not\equiv 1\pmod{p}$ and $\eta_{1}(t_{1})=\eta_{1}(t_{2})$, where $t_{1}=\phi(\ell^{k})\pmod{p}$ and $t_{2}=\ell^{k-1}\pmod{p}$ & $[p^{2m-1}+\epsilon_{f}\sqrt{p^{*}}^{m}\left(p^{m-1}(p-1)-\frac{p^{m}-1}{\ell^{k-1}}\right)-1,2m]$ & $t=8$ & Table \ref{Table3} \\
              & $p\equiv 1\pmod{4}$, $\ell\not\equiv 1\pmod{p}$ and $\eta_{1}(t_{1})=-\eta_{1}(t_{2})$, where $t_{1}=\phi(\ell^{k})\pmod{p}$ and $t_{2}=\ell^{k-1}\pmod{p}$ & $[p^{2m-1}+\epsilon_{f}\sqrt{p^{*}}^{m}\left(p^{m-1}(p-1)-\frac{p^{m}-1}{\ell^{k-1}}\right)-1,2m]$ & $t=9$ & Table \ref{Table4} \\
               & $p\equiv 3\pmod{4}$ and $\ell\not\equiv 1\pmod{p}$ & $[p^{2m-1}+\epsilon_{f}\sqrt{p^{*}}^{m}\left(p^{m-1}(p-1)-\frac{p^{m}-1}{\ell^{k-1}}\right)-1,2m]$ & $t=6$ & Table \ref{Table5} \\
            \hline 
        \end{tabular}
    \end{adjustbox}
    \caption{Known $p$-ary linear code $\mathcal{C}_{D}$ with $t$-weights of the form $(\ref{EQ2})$ by using the defining set $D\subset\mathbb{F}_{p^{m}}^{2}$, where $p$ is an odd prime.}
    \label{Table7}
\end{table}
\section{Preliminaries and auxiliary lemmas}\label{sec2}
In this section, we present some preliminaries and important results related to exponential sums over finite fields, cyclotomic fields, and weakly regular bent functions, which will be used in the subsequent sections.

\subsection{Exponential sums}\label{ssec2.1}
Let $m\geq 1$ be an integer. A canonical additive character $\chi_{m}$ of $\mathbb{F}_{p^m}$ is a homomorphism from the additive group $\mathbb{F}_{p^m}$ to the set of complex numbers with absolute value $1$. For each $x\in\mathbb{F}_{p^m}$, $\chi_{m}$ is defined by
\begin{equation*}
    \chi_{m}(x)=\zeta_{p}^{\operatorname{Tr}_{1}^{m}(x)},
\end{equation*}
 where $\zeta_{p}=e^{\frac{2\pi \sqrt{-1}}{p}}$ is a primitive $p$-th root of unity. For a fixed primitive element $g$ of $\mathbb{F}_{p^m}$, the quadratic multiplicative character $\eta_{m}$ of $\mathbb{F}_{p^m}$ is defined as $\eta_{m}(g^{k})=e^{k\pi\sqrt{-1}}$, where $k=0,1,\cdots,p^m-2$. From \cite[Lemma $7$]{P15}, it is known that for any $x\in\mathbb{F}_{p}^{*}$, $\eta_{m}(x)=1$, if $m$ is even and $\eta_{m}(x)=\eta_{1}(x)$, if $m$ is odd. We extend this quadratic character by setting $\eta_{m}(0)=0$. The quadratic Gauss sum $G(\eta_{m},\chi_{m})$ over $\mathbb{F}_{p^m}$ is defined by
\begin{equation*}
    G(\eta_{m},\chi_{m})=\sum_{c\in\mathbb{F}_{p^m}^{*}}\eta_{m}(c)\chi_{m}(c).
\end{equation*}
The notation $\eta_{1}$ and $\chi_{1}$ denotes the quadratic and canonical additive characters of $\mathbb{F}_{p}$, respectively.  
\vspace{1mm}
\begin{lemma}\textnormal{\cite[Theorem 5.15]{PP15}}
    Let $G(\eta_{m},\chi_{m})$ be defined as above; then
    \begin{equation*}
        G(\eta_{m},\chi_{m})=(-1)^{m-1}\sqrt{(p^{*})^{m}}=\begin{cases}
            (-1)^{m-1}p^{\frac{m}{2}};\text{ if }p\equiv 1\pmod{4}, \\
            (-1)^{m-1}(\sqrt{-1})^{m}p^{\frac{m}{2}};\text{ if }p\equiv 3\pmod{4},
        \end{cases}
    \end{equation*}
    where $p^{*}=\eta_{1}(-1)p=(-1)^{\frac{p-1}{2}}p$.
\end{lemma}
\vspace{1mm}
\begin{lemma}\textnormal{\cite[Theorem 5.33]{PP15}}\label{Le6}
    Let $f(x)=r_{2}x^{2}+r_{1}x+r_{0}\in\mathbb{F}_{p^m}[x]$ with $r_{2}\neq 0$. Then
    \begin{equation*}
        \sum_{x\in\mathbb{F}_{p^m}}\chi_{m}(f(x))=\chi_{m}(r_{0}-r_{1}^{2}(4r_{2})^{-1})\eta_{m}(r_{2})G(\eta_{m},\chi_{m}).
    \end{equation*}
\end{lemma}
Throughout we assume that $\alpha$ be a fixed primitive element of $\mathbb{F}_{q}$, $q=p^{e}$, $e=\phi(\ell^{k})$, $p$ is a primitive root modulo $2\ell^{k}$, $\operatorname{Ind}_{\alpha}(b)$ be the index of $b\in\mathbb{F}_{q}^{*}$ with respect to the base $\alpha$,  and $\xi=\alpha^{\frac{q-1}{2\ell^{k}}}$. One can check that $\mathbb{F}_{q}=\mathbb{F}_{p}(\xi)$. It follows that $\{\xi,\xi^2,\cdots,\xi^{e}\}$ forms a basis of $\mathbb{F}_{q}$ over $\mathbb{F}_{p}$. In \cite{P8}, Cheng and Gao defined the following two exponential sums over $\mathbb{F}_{q}$ given by 
\begin{align}
  \hskip 30pt S_{2\ell^{k}}(a,b) &:=\sum_{x\in\mathbb{F}_{q}^{*}}\chi_{e}(ax^{\frac{q-1}{2\ell^{k}}}+bx), & \label{EQ3}
  \end{align}
   and 
    \begin{align}
  S_{2\ell^{k}}(a) &:=\sum_{i=0}^{2\ell^{k}-1}\chi_{e}(a\xi^{i}),\text{ where $a,b\in\mathbb{F}_{q}$.} & \label{EQ4}
\end{align}
Define a set $P=\{i\in \mathbb{Z}:0\leq i\leq 2\ell^{k}-1\}$. Let $P_{1}, P_{2}, P_{3}$, and $P_{4}$ be the subsets of $P$ defined as follows: 
\begin{align*}
    P_{1} &=\{i\in P: 1\leq i \leq (\ell-1)\ell^{k-1}\}, \\
    P_{2} &= \{i\in P: (\ell-1)\ell^{k-1}< i \leq \ell^{k}\}, \\
    P_{3} &= \{i\in P: \ell^{k}< i \leq 2\ell^{k}-\ell^{k-1}\}, \\
    P_{4} &= \{i\in P: 2\ell^{k}-\ell^{k-1}< i \leq 2\ell^{k}-1\}.
\end{align*} 
It is clear that $(\{0\},P_{1},P_{2},P_{3},P_{4})$ is a partition of $P$. Now we further partition the two sets $P_{1}$ and $P_{3}$ into the union of three disjoint subsets, $P_{1}=P_{1}^{(1)}\cup P_{1}^{(2)}\cup P_{1}^{(3)}$ and $P_{3}=\{\ell^{k}+u\ell^{k-1}-v: 1\leq u\leq \ell-1\text{ and }0\leq v \leq \ell^{k-1}-1\}=P_{3}^{(1)}\cup P_{3}^{(2)}\cup P_{3}^{(3)}$, respectively, where 
\begin{align*}
    P_{1}^{(1)}&= \{i\in P_{1}:\ell^{k-1}\nmid i\},\text{ }P_{1}^{(2)}=\{i\in P_{1}:\ell^{k-1}\mid i\text{ and }2\nmid i\},\text{ }
    P_{1}^{(3)}=\{i\in P_{1}:2\ell^{k-1}\mid i\}; &
\end{align*}
and
\begin{align*}
    P_{3}^{(1)}&= \{i\in P_{3}:\ell^{k}+u\ell^{k-1}-v: 1\leq u\leq \ell-1\text{ and }0<v\leq\ell^{k-1}-1\},  &\\
    P_{3}^{(2)}& =\{i\in P_{3}:\ell^{k}+u\ell^{k-1}: 1\leq u\leq \ell-1\text{ and }2\mid u\},   &\\
    P_{3}^{(3)}&=\{i\in P_{3}:\ell^{k}+u\ell^{k-1}: 1\leq u\leq \ell-1\text{ and }2\nmid u\}.  &
\end{align*}

The following fundamental results from \cite{P8, P20} will be useful in the subsequent sections.
\vspace{1mm}
\begin{lemma}\label{Le2}
    If $S_{2\ell^{k}}(a,b)$ and $S_{2\ell^{k}}(a)$ are defined as in $(\ref{EQ3})$ and $(\ref{EQ4})$ respectively. Then, we have
    \begin{enumerate}
    \vspace{1mm}
        \item[\textnormal{(a)}] For any $a\in\mathbb{F}_{q}$, we have
        \begin{equation*}
            S_{2\ell^{k}}(a,0)=\frac{q-1}{2\ell^{k}}S_{2\ell^{k}}(a).
        \end{equation*}
         \item[\textnormal{(b)}] For any $a\in\mathbb{F}_{p}$, we have 
       \begin{equation*}
           S_{2\ell^{k}}(a)=\zeta_{p}^{\ell^{k-1}(\ell-1)a}+\zeta_{p}^{-\ell^{k-1}(\ell-1)a}+(\ell-1)(\zeta_{p}^{\ell^{k-1}a}+\zeta_{p}^{-\ell^{k-1}a})+2\ell^{k}-2\ell.
       \end{equation*} 
       \item[\textnormal{(c)}] For every $i\in P=\{i\in \mathbb{Z}:0\leq i\leq 2\ell^{k}-1\}$ and $a\in\mathbb{F}_{p}$, we have
       \begin{equation*}
           S_{2\ell^{k}}(a\xi^{i})=S_{2\ell^{k}}(a).
       \end{equation*}
       \vspace{1mm}
        \item[\textnormal{(d)}] For any $a\in\mathbb{F}_{q}$ and $b\in\mathbb{F}_{q}^{*}$, we have 
        \begin{equation*}
            S_{2\ell^{k}}(a,b)=\chi_{e}(c)\sqrt{q}-\frac{\sqrt{q}+1}{2\ell^{k}}S_{2\ell^{k}}(c),\text{ where, $c=ab^{-\frac{q-1}{2\ell^{k}}}$}
        \end{equation*}
        \vspace{1mm}
      \item[\textnormal{(e)}]  For any $a\in \mathbb{F}_{q}$ and $y\in\mathbb{F}_{p}^{*}$, we have $S_{2\ell^{k}}(a,b)=S_{2\ell^{k}}(a,yb)$. 
    \end{enumerate}
\end{lemma}
\vspace{1mm}
\begin{lemma}\textnormal{\cite[Lemma 2.8]{P8}}\label{Le3}
    Let $i\in P=\{i\in \mathbb{Z}:0\leq i\leq 2\ell^{k}-1\}$. Then
    \begin{align*}
        \operatorname{Tr}_{1}^{e}(\xi^{i}) &= \begin{cases}
            (\ell-1)\ell^{k-1},\text{ if }i=0, \\
            -(\ell-1)\ell^{k-1},\text{ if }i=\ell^{k}, \\
            -\ell^{k-1},\text{ if }i\in P_{1}^{(3)}\cup P_{3}^{(3)}, \\
            \ell^{k-1},\text{ if }i\in P_{1}^{(2)}\cup P_{3}^{(2)}, \\
            0,\text{ otherwise.}
        \end{cases}
    \end{align*}
\end{lemma}
\vspace{1mm}
\begin{remark}
    If $b\in\mathbb{F}_{q}^{*}$ and let $i_{b}=-\operatorname{Ind}_{\alpha}(b)\pmod{2\ell^{k}}$, then $b^{-\frac{q-1}{2\ell^{k}}}=\alpha^{-\frac{q-1}{2\ell^{k}}\operatorname{Ind}_{\alpha}(b)}=\xi^{i_{b}}$. Clearly, $i_{b}\in P=\{i\in \mathbb{Z}:0\leq i\leq 2\ell^{k}-1\}$ and $\chi_{e}(z b^{-\frac{q-1}{2\ell^{k}}})=\zeta_{p}^{z\operatorname{Tr}_{1}^{e}(\xi^{i_{b}})}$ for all $z\in\mathbb{F}_{p}^{*}$. The value of $\operatorname{Tr}_{1}^{e}(\xi^{i_{b}})$ can be calculated from Lemma $\ref{Le3}$. Hence, one can evaluate $S_{2\ell^{k}}(z,b)$ for $z\in\mathbb{F}_{p}^{*}$ and $b\in\mathbb{F}_{q}^{*}$.
\end{remark}

 \vspace{1em}
 Define the following useful notation $w(u,b)$ as follows:
    \begin{equation}\label{EQ22}
        w(u,b)=\sum_{z\in\mathbb{F}_{p}^{*}}\zeta_{p}^{-uz}S_{2\ell^{k}}(z,b),\text{ where $u\in\mathbb{F}_{p}$ and $b\in\mathbb{F}_{q}$.}
    \end{equation}
\vspace{1em}
\begin{lemma}\textnormal{\cite[Proposition 3.2]{P20}}\label{Le5}
    For $b\in\mathbb{F}_{q}^{*}$, let $i_{b}=-\operatorname{Ind}_{\alpha}(b)\pmod{2\ell^{k}}$. Suppose $w(u,b)$ is defined as in $(\ref{EQ22})$.  Then the explicit values of $w(u,b)$ are as follows:
\begin{enumerate}
    \item[\textnormal{(a)}] When $u=b=0$. Then we have
    \begin{flalign*}
   \hskip 15pt     w(0,0) &= \begin{cases}
        \frac{q-1}{\ell^{k}}\left((p-1)\ell^{k}-p\ell+p\right);\text{ if }\ell\equiv 1\pmod{p}, \\
        \frac{q-1}{\ell^{k-1}}\left((p-1)\ell^{k-1}-p\right);\text{ if }\ell\not\equiv 1\pmod{p}. 
        \end{cases}   &
    \end{flalign*}
    \item[\textnormal{(b)}] When $u=0$ and $b\neq 0$. Then we have
    \begin{flalign*}
  \hskip 15pt      w(0,b) &=\begin{cases}
            \frac{\sqrt{q}+1}{\ell^{k}}(-p\ell^{k}+p\ell-p)+1;\text{ if } i_{b}\in P_{1}^{(2)}\cup P_{1}^{(3)}\cup P_{3}^{(2)}\cup P_{3}^{(3)}, \\
                \frac{\sqrt{q}+1}{\ell^{k}}(p\ell-p)-p+1;\text{ otherwise},
        \end{cases} &
    \end{flalign*}
    if $\ell\equiv1\pmod{p}$; and 
     \begin{flalign*}
  \hskip 15pt      w(0,b) &=\begin{cases}
             \frac{\sqrt{q}+1}{\ell^{k-1}}(-p\ell^{k-1}+p)+1;\text{ if }i_{b}\in\{0,\ell^{k}\}\cup P_{1}^{(2)}\cup P_{1}^{(3)}\cup P_{3}^{(2)}\cup P_{3}^{(3)}, \\
                \frac{p(\sqrt{q}+1)}{\ell^{k-1}}-p+1;\text{ otherwise},
        \end{cases} &
    \end{flalign*}
    if $\ell\not\equiv1\pmod{p}$.
    \vspace{2mm}
    \item[\textnormal{(c)}] When $u\neq 0$ and $b=0$. Then we have
    \begin{flalign*}
 \hskip 15pt       w(u,0) &= \begin{cases}
     \frac{q-1}{2\ell^{k}}\left(-2\ell^{k}+p\right);\text{ if }u^{2}\equiv\phi(\ell^{k})^{2}\pmod{p}\text{ but }u^{2}\not\equiv\ell^{2(k-1)}\pmod{p}, \\
        \frac{q-1}{2\ell^{k-1}}\left(-2\ell^{k-1}+p\right);\text{ if }\text{ }u\equiv\phi(\ell^{k})\equiv\ell^{k-1}\pmod{p}\text{ or }\\\hspace{3.5cm}u\equiv-\phi(\ell^{k})\equiv-\ell^{k-1}\pmod{p}, \\
        \frac{q-1}{2\ell^{k}}(-2\ell^{k}+\ell p-p);\text{ if }u^{2}\not\equiv\phi(\ell^{k})^{2}\pmod{p}\text{ but }
        u^{2}\equiv \ell^{2(k-1)}\pmod{p}, \\
        1-q;\text{ if }u^{2}\not\equiv\phi(\ell^{k})^{2}\text{ and }u^{2}\not\equiv\ell^{2(k-1)}\pmod{p}.
 \end{cases}
    \end{flalign*}
    \item[\textnormal{(d)}] When $u\neq 0$ and $b\neq 0$. Then $w(u,b)=1$ if $u^{2}\not\equiv\phi(\ell^{k})^{2}\pmod{p}$ and $u^{2}\not\equiv\ell^{2(k-1)}\pmod{p}$. Moreover, we have the following results.
    \begin{enumerate}
        \item[\textnormal{(i)}]  \begin{flalign*}
        \hskip 15pt  w(u,b) &=\begin{cases}
             1+p\sqrt{q}-p\left(\frac{\sqrt{q}+1}{2\ell^{k}}\right);\text{ if }i_{b}=0, \\
              1-p\left(\frac{\sqrt{q}+1}{2\ell^{k}}\right); \text{ otherwise}, 
        \end{cases} &
    \end{flalign*}
    if $u\equiv\phi(\ell^{k})\pmod{p}$ with $u^{2}\not\equiv\ell^{2(k-1)}\pmod{p}$;
    
   \item[\textnormal{(ii)}]
    \begin{flalign*}
        \hskip 15pt  w(u,b) &=\begin{cases}
              1+p\sqrt{q}-p\left(\frac{\sqrt{q}+1}{2\ell^{k}}\right);\text{ if }i_{b}=\ell^{k}, \\
              1-p\left(\frac{\sqrt{q}+1}{2\ell^{k}}\right); \text{ otherwise}, 
        \end{cases} &
    \end{flalign*}
    if $u\equiv-\phi(\ell^{k})\pmod{p}$ with $u^{2}\not\equiv\ell^{2(k-1)}\pmod{p}$;
    \item[\textnormal{(iii)}]
    \begin{flalign*}
        \hskip 15pt  w(u,b) &=\begin{cases}
              1+p\sqrt{q}-p(\ell-1)\left(\frac{\sqrt{q}+1}{2\ell^{k}}\right);\text{ if }i_{b}\in P_{1}^{(2)}\cup P_{3}^{(2)}, \\
              1-p(\ell-1)\left(\frac{\sqrt{q}+1}{2\ell^{k}}\right); \text{ otherwise}, 
        \end{cases} &
    \end{flalign*}
    if $u\equiv\ell^{k-1}\pmod{p}$ with $u^{2}\not\equiv\phi(\ell^{k})^{2}\pmod{p}$;
    \item[\textnormal{(iv)}]
    \begin{flalign*}
        \hskip 15pt  w(u,b) &=\begin{cases}
              1+p\sqrt{q}-p(\ell-1)\left(\frac{\sqrt{q}+1}{2\ell^{k}}\right);\text{ if }i_{b}\in P_{1}^{(3)}\cup P_{3}^{(3)}, \\
              1-p(\ell-1)\left(\frac{\sqrt{q}+1}{2\ell^{k}}\right); \text{ otherwise}, 
        \end{cases} &
    \end{flalign*}
    if $u\equiv-\ell^{k-1}\pmod{p}$ with $u^{2}\not\equiv\phi(\ell^{k})^{2}\pmod{p}$;
    \item[\textnormal{(v)}]
    \begin{flalign*}
        \hskip 15pt  w(u,b) &=\begin{cases}
              1+p\sqrt{q}-p\left(\frac{\sqrt{q}+1}{2\ell^{k-1}}\right);\text{ if }i_{b}\in \{0\}\cup P_{1}^{(2)}\cup P_{3}^{(2)}, \\
              1-p\left(\frac{\sqrt{q}+1}{2\ell^{k-1}}\right); \text{ otherwise}, 
        \end{cases} &
    \end{flalign*}
    if $u\equiv\phi(\ell^{k})\equiv\ell^{k-1}\pmod{p}$; and
    \item[\textnormal{(vi)}]
    \begin{flalign*}
        \hskip 15pt  w(u,b) &=\begin{cases}
              1+p\sqrt{q}-p\left(\frac{\sqrt{q}+1}{2\ell^{k-1}}\right);\text{ if }i_{b}\in \{\ell^{k}\}\cup P_{1}^{(3)}\cup P_{3}^{(3)}, \\
              1-p\left(\frac{\sqrt{q}+1}{2\ell^{k-1}}\right); \text{ otherwise}, 
        \end{cases} &
    \end{flalign*}
    if $u\equiv-\phi(\ell^{k})\equiv-\ell^{k-1}\pmod{p}$.
\end{enumerate}
\end{enumerate}
\end{lemma}
 \vspace{1mm}
   It is to be noted that when $u\neq 0$, then $u\equiv\phi(\ell^{k})\equiv\ell^{k-1}\pmod{p}$ would imply $u\not\equiv-\phi(\ell^{k})\pmod{p}$ and $u\not\equiv-\ell^{k-1}\pmod{p}$; and $u\equiv-\phi(\ell^{k})\equiv-\ell^{k-1}\pmod{p}$ would imply $u\not\equiv\phi(\ell^{k})\pmod{p}$ and $u\not\equiv\ell^{k-1}\pmod{p}$.
\subsection{Cyclotomic fields $\mathbb{Q}(\zeta_{p})$}
Let $\mathbb{Q}$ be the field of rational numbers and $\mathbb{Z}$ be the set of integers. The cyclotomic field $\mathbb{Q}(\zeta_{p})$ is obtained from the field $\mathbb{Q}$ by adjoining the $p$-th primitive root of unity $\zeta_{p}$. 
\vspace{1mm}
\begin{lemma}$\textnormal{\cite{R8}}$\label{Le1}
    Let $F=\mathbb{Q}(\zeta_{p})$ be the $p$-th cyclotomic field over $\mathbb{Q}$. Then 
    \vspace{2mm}
    \begin{enumerate}
        \item[$(1)$]  The ring of integers in $F=\mathbb{Q}(\zeta_{p})$ is defined as $\mathbb{Z}(\zeta_{p})$, and an integral basis of $\mathbb{Z}(\zeta_{p})$ is the set $\{\zeta_{p}^{i}:1\leq i\leq p-1\}$.
       \vspace{2mm}
        \item[$(2)$]  The field extension $F/\mathbb{Q}$ is Galois of degree $p-1$ and Galois group $Gal(F/\mathbb{Q})=\{\sigma_{z}:z\in\mathbb{F}_{p}^{*}\}$, where the automorphism $\sigma_{z}$ of $F$ is defined by $\sigma_{z}(\zeta_{p})=\zeta_{p}^{z}$.
        \vspace{2mm}
        \item[$(3)$] The field $F$ has a unique quadratic subfield $L=\mathbb{Q}(\sqrt{p^{*}})$, where $p^{*}=\eta_{1}(-1)p$, $\eta_{1}$ is the quadratic character over $\mathbb{F}_{p}^{*}$. For $z\in\mathbb{F}_{p}^{*}$, $\sigma_{z}(\sqrt{p^{*}})=\eta_{1}(z)\sqrt{p^{*}}$. Therefore, the Galois group $Gal(L/\mathbb{Q})$ is $\{1,\sigma_{w}\}$, where $w$ is any quadratic non-residue in $\mathbb{F}_{p}$.
    \end{enumerate}
\end{lemma}
\vskip 1pt
By Lemma \ref{Le1}, we get, $\sqrt{p^{*}}^{m}=\eta_{1}^{\frac{m}{2}}(-1)p^{\frac{m}{2}}$, $\sigma_{z}(\zeta_{p}^{y})=\zeta_{p}^{zy}$ and $\sigma_{z}(\sqrt{p^{*}}^{m})=\eta_{1}^{m}(z)\sqrt{p^{*}}^{m}=\sqrt{p^{*}}^{m}$, if $m$ is an even positive integer.

\subsection{Weakly regular bent functions}\label{S1}

 Let $f:\mathbb{F}_{p^m}\rightarrow\mathbb{F}_{p}$ be a $p$-ary function. The Walsh transform of $f$ is defined by 
 \begin{equation*}
     \mathcal{W}_{f}(\lambda)=\sum_{x\in\mathbb{F}_{q}}\zeta_{p}^{f(x)-\operatorname{Tr}_{1}^{m}(\lambda x)},
 \end{equation*}
  where $\lambda\in\mathbb{F}_{p^m}$ and $\zeta_{p}$ is a complex primitive $p$-th root of unity. The $p$-ary function $f$ is said to be bent if $|\mathcal{W}_{f}(\lambda)|^{2}=p^{m}$ for all $\lambda\in\mathbb{F}_{p^m}$. A bent function $f$ is said to be regular bent if there exists some $p$-ary function $f^{*}:\mathbb{F}_{p^m}\rightarrow\mathbb{F}_{p}$ such that $\mathcal{W}_{f}(\lambda)=p^{m/2}\zeta_{p}^{f^{*}(\lambda)}$ for every $\lambda\in\mathbb{F}_{p^m}$. The bent function $f$ is called weakly regular bent if there exists a $p$-ary function $f^{*}$ and a complex number $u$ with unit magnitude satisfying $\mathcal{W}_{f}(\lambda)=up^{\frac{m}{2}}\zeta_{p}^{f^{*}(\lambda)}$ for every $\lambda\in\mathbb{F}_{p^m}$. Such a function $f^{*}$ is called the dual of $f$. The dual of a weakly regular bent function is also weakly regular bent. From \cite{S3}, a weakly regular bent function $f(x)$ satisfies $\mathcal{W}_{f}(\lambda)=\epsilon_{f}\sqrt{p^{*}}^{m}\zeta_{p}^{f^{*}(\lambda)}$, where $\epsilon_{f}\in\{\pm 1\}$ is called the sign of the Walsh transform of $f(x)$ and $p^{*}=(-1)^{\frac{p-1}{2}}p$, and the sign of the Walsh transform of $f^{*}(x)$ is $(-1)^{(\frac{p-1}{2})m}\epsilon_{f}$. 
 We denote $\mathcal{RF}$ to be the set of $p$-ary weakly regular bent functions $f(x)$ with $f(0)=0$ and $f(cx)=c^{k_{f}}f(x)$ for any $c\in\mathbb{F}_{p}^{*}$, where $k_{f}$ is a positive even integer with $\operatorname{gcd}(k_{f}-1,p-1)=1$. 
 It is well known that $f^{*}(0)=0$ when $f\in\mathcal{RF}$, and there is an even positive integer $l_{f}$ satisfying $\operatorname{gcd}(l_{f}-1,p-1)=1$ such that $f^{*}(cx)=c^{l_{f}}f^{*}(x)$ for any $c\in\mathbb{F}_{p}^{*}$ and $x\in\mathbb{F}_{p^m}$ (see \cite[Proposition $2.4$ and $2.5$]{P18}). Weakly regular bent functions can be constructed from known planar functions. A function $g$ from $\mathbb{F}_{p^{m}}$ to $\mathbb{F}_{p^{m}}$ is said to be planar if $|\{x\in \mathbb{F}_{p^{m}}: g(x+a)-g(x)=b\}|=1$ for any $a\in\mathbb{F}_{p^m}^{*}$ and $b\in\mathbb{F}_{p^m}$. It is known that the $p$-ary functions $f(x)=\operatorname{Tr}_{1}^{m}(\alpha g(x))$, $\alpha\in\mathbb{F}_{p^m}$ for which $g(x)=\sum_{0\leq i,j\leq m-1}a_{ij}x^{p^{i}+p^{j}}$ are (DO-type) planar, are all weakly regular quadratic bent functions, when $m$ is even. For more information, we refer to \cite{S1,S2,S3}. 
 \vskip 1pt
 In this paper, we mainly focus on those functions $f\in\mathcal{RF}$ such that $f^{*}$ is a quadratic form (i.e., $l_{f}=2$). The known weakly regular bent functions satisfying $l_{f}=2$ are provided in \cite[Table 4]{P12}. 
\vspace{1mm}
\vskip 1pt
\begin{lemma}\textnormal{\cite[Lemma 8]{P18}}\label{LX3}
    Let $f(x)\in\mathcal{RF}$ such that $\mathcal{W}_{f}(0)=\epsilon_{f}\sqrt{p^{*}}^{m}$, where $\epsilon_{f}\in\{1,-1\}$. Let $f^{*}$ be the dual of $f$. Define $N_{f^{*}}(\lambda)=|\{a\in\mathbb{F}_{p^{m}}:f^{*}(a)=\lambda\}|$, where $\lambda\in\mathbb{F}_{p}$. Then we have the following.  
    \begin{enumerate}
    \vspace{1mm}
        \item[\textnormal{(1)}] When $m$ is even, we have 
        \begin{align*}
            N_{f^{*}}(\lambda)=\begin{cases}
                p^{m-1}+\epsilon_{f}(p-1)\frac{\sqrt{p^{*}}^{m}}{p};\text{ if }\lambda=0, \\
                p^{m-1}-\epsilon_{f}\frac{\sqrt{p^{*}}^{m}}{p};\text{ if }\lambda\in\mathbb{F}_{p}^{*}.
            \end{cases}
        \end{align*}
        \item[\textnormal{(2)}] When $m$ is odd, we have
        \begin{align*}
            N_{f^{*}}(\lambda)=\begin{cases}
                p^{m-1};\text{ if }\lambda=0, \\
                p^{m-1}+\epsilon_{f}\eta_{1}(-1)\sqrt{p^{*}}^{m-1};\text{ if }\eta_{1}(\lambda)=1, \\
                p^{m-1}-\epsilon_{f}\eta_{1}(-1)\sqrt{p^{*}}^{m-1};\text{ if }\eta_{1}(\lambda)=-1.
            \end{cases}
        \end{align*}
    \end{enumerate}
\end{lemma}
 \section{Linear codes from the defining set $D_{u}$}\label{Sec3}
 In this section, we consider the linear code $\mathcal{C}_{D_{u}}$ defined in $(\ref{EQ2})$ and $(\ref{EQ9})$. We need the following lemmas for determining the lengths and weight distributions of $\mathcal{C}_{D_{u}}$.
 \vspace{1mm}
 \begin{lemma}\label{Le8}
     If $D_{u}=\{(x,y)\in\mathbb{F}_{q}^{2}\backslash\{(0,0)\}:\operatorname{Tr}_{1}^{e}(x+y^{\frac{q-1}{2\ell^{k}}})=u\}$, where $u\in\mathbb{F}_{p}$. Let $n_{1}=|D_{u}|$. Then \begin{align*}
  \hskip 30pt    n_{1}&= \begin{cases}
         p^{2e-1}-1;\text{ if }u=0, \\
         p^{2e-1};\text{ if }u\in\mathbb{F}_{p}^{*}.
             \end{cases}  &
     \end{align*} 
 \end{lemma}
 \begin{proof}
 From the properties of the canonical additive characters and using the fact $\sum_{x\in\mathbb{F}_{q}}\zeta_{p}^{z\operatorname{Tr}_{1}^{e}(x)}=0$ for all $z\in\mathbb{F}_{p}^{*}$, we obtain
      \begin{align*}
     n_{1} &=\frac{1}{p}\sum_{x,y\in\mathbb{F}_{q}}\sum_{z\in\mathbb{F}_{p}}\zeta_{p}^{z\left(\operatorname{Tr}_{1}^{e}(x+y^{\frac{q-1}{2\ell^{k}}})-u\right)}-\frac{1}{p}\sum_{z\in\mathbb{F}_{p}}\zeta_{p}^{-uz} \\
     &=\frac{1}{p}\left(q^{2}+\sum_{z\in\mathbb{F}_{p}^{*}}\zeta_{p}^{-uz}\sum_{x\in\mathbb{F}_{q}}\zeta_{p}^{z\operatorname{Tr}_{1}^{e}(x)}\sum_{y\in\mathbb{F}_{q}}\zeta_{p}^{z\operatorname{Tr}_{1}^{e}(y^{\frac{q-1}{2\ell^{k}}})}\right)-\frac{1}{p}\sum_{z\in\mathbb{F}_{p}}\zeta_{p}^{-uz} \\
     &= p^{2e-1}-\frac{1}{p}\sum_{z\in\mathbb{F}_{p}}\zeta_{p}^{-uz} \\
     &= \begin{cases}
         p^{2e-1}-1;\text{ if }u=0, \\
         p^{2e-1};\text{ if }u\in\mathbb{F}_{p}^{*}.
     \end{cases}
 \end{align*}
 \end{proof}
Define a notation 
\begin{equation}\label{EQ5}
    N_{\gamma,\delta}^{(1)}=|\{(x,y)\in D_{u}: \operatorname{Tr}_{1}^{e}(\gamma x+\delta y)=0\}|,
\end{equation}
and the Hamming weight of the codeword $c(\gamma,\delta)\in \mathcal{C}_{D_{u}}$ is as follows:
\begin{equation}\label{EQ11}
    \operatorname{wt}(c(\gamma,\delta))=n_{1}-N_{\gamma,\delta}^{(1)}.
\end{equation}
\begin{lemma}\label{EQ12}
    Let $\gamma,\delta\in\mathbb{F}_{q}$ with $(\gamma,\delta)\neq (0,0)$ and $N_{\gamma,\delta}^{(1)}$ be defined in $(\ref{EQ5})$. For $b\in\mathbb{F}_{q}^{*}$, let $\operatorname{Ind}_{\alpha}(b)$ be the index of $b$ with respect to the base $\alpha$ and $i_{b}=-\operatorname{Ind}_{\alpha}(b)\pmod{2\ell^{k}}$. Assume that the set $P$ and its partitions are defined as in section $\ref{ssec2.1}$. Then the explicit value of $N_{\gamma,\delta}^{(1)}$ is given as follows:
    \begin{enumerate} 
        \item[\textnormal{(a)}] If $u=0$, then
        \begin{enumerate}
            \item[\textnormal{(i)}] When $\ell\equiv 1\pmod{p}$, then
            \begin{flalign*}
                N_{\gamma,\delta}^{(1)} &=\begin{cases}
                    p^{2e-1}-1-p^{e-1}(\ell-1)\frac{q-1}{\ell^{k}};\text{ if }\gamma\in\mathbb{F}_{p}^{*}\text{ and }\delta=0, \\
                    p^{2e-2}-1+p^{e-1}(\ell-1)\frac{\sqrt{q}+1}{\ell^{k}}-p^{e-1}\sqrt{q};\text{ if }\gamma\in\mathbb{F}_{p}^{*}\text{ and }\delta\neq 0\text{ with }\\\hspace{5.4cm} i_{\delta}\in P_{1}^{(2)}\cup P_{1}^{(3)}\cup P_{3}^{(2)}\cup P_{3}^{(3)}, \\
                    p^{2e-2}-1+p^{e-1}(\ell-1)\frac{\sqrt{q}+1}{\ell^{k}};\text{ if }\gamma\in\mathbb{F}_{p}^{*}\text{ and }\delta\neq 0\text{ with }\\\hspace{3.1cm} i_{\delta}\in P\backslash(P_{1}^{(2)}\cup P_{1}^{(3)}\cup P_{3}^{(2)}\cup P_{3}^{(3)}), \\
                    p^{2e-2}-1;\text{ otherwise}.
                \end{cases} &
            \end{flalign*}
            \item[\textnormal{(ii)}] When $\ell\not\equiv 1\pmod{p}$, then
            \begin{flalign*}
                N_{\gamma,\delta}^{(1)} &=\begin{cases}
                    p^{2e-1}-1-p^{e-1}\frac{q-1}{\ell^{k-1}};\text{ if }\gamma\in\mathbb{F}_{p}^{*}\text{ and }\delta=0, \\
                    p^{2e-2}-1+p^{e-1}\frac{\sqrt{q}+1}{\ell^{k-1}}-p^{e-1}\sqrt{q};\text{ if }\gamma\in\mathbb{F}_{p}^{*}\text{ and }\delta\neq 0\text{ with }\\ \hspace{5.2cm}i_{\delta}\in \{0,\ell^{k}\}\cup P_{1}^{(2)}\cup P_{1}^{(3)}\cup P_{3}^{(2)}\cup P_{3}^{(3)}, \\
                    p^{2e-2}-1+p^{e-1}\frac{\sqrt{q}+1}{\ell^{k-1}};\text{ if }\gamma\in\mathbb{F}_{p}^{*}\text{ and }\delta\neq 0\text{ with }\\ \hspace{3.68cm}i_{\delta}\in P\backslash(\{0,\ell^{k}\}\cup P_{1}^{(2)}\cup P_{1}^{(3)}\cup P_{3}^{(2)}\cup P_{3}^{(3)}), \\
                    p^{2e-2}-1;\text{ otherwise}.
                \end{cases} &
            \end{flalign*}
        \end{enumerate}
        \item[\textnormal{(b)}] If $u\neq 0$, then
        \begin{enumerate}
             \item[\textnormal{(i)}] When $u^{2}\not\equiv\phi(\ell^{k})^{2}\pmod{p}$ and $u^{2}\not\equiv\ell^{2(k-1)}\pmod{p}$, then
        \begin{flalign*}
            N_{\gamma,\delta}^{(1)} &=\begin{cases}
                0;\text{ if }\gamma\in\mathbb{F}_{p}^{*}\text{ and }\delta=0, \\
                p^{2e-2};\text{ otherwise}. 
            \end{cases}  &
        \end{flalign*}
        \item[\textnormal{(ii)}] When $u\equiv\phi(\ell^{k})\pmod{p}$ and $u^{2}\not\equiv\ell^{2(k-1)}\pmod{p}$, then
        \begin{flalign*}
            N_{\gamma,\delta}^{(1)} &=\begin{cases}
                p^{e-1}\frac{q-1}{2\ell^{k}};\text{ if }\gamma\in\mathbb{F}_{p}^{*}\text{ and }\delta=0, \\
                p^{2e-2}+p^{e-1}\sqrt{q}-p^{e-1}\frac{\sqrt{q}+1}{2\ell^{k}};\text{ if }\gamma\in\mathbb{F}_{p}^{*}\text{ and }\delta\neq 0\text{ with }i_{\delta}=0, \\
                p^{2e-2}-p^{e-1}\frac{\sqrt{q}+1}{2\ell^{k}};\text{ if }\gamma\in\mathbb{F}_{p}^{*}\text{ and }\delta\neq 0\text{ with }i_{\delta}\in P\backslash\{0\}, \\
                p^{2e-2};\text{ otherwise}.
            \end{cases}  &
        \end{flalign*}
         \item[\textnormal{(iii)}] When $u\equiv-\phi(\ell^{k})\pmod{p}$ and $u^{2}\not\equiv\ell^{2(k-1)}\pmod{p}$, then
        \begin{flalign*}
            N_{\gamma,\delta}^{(1)} &=\begin{cases}
                p^{e-1}\frac{q-1}{2\ell^{k}};\text{ if }\gamma\in\mathbb{F}_{p}^{*}\text{ and }\delta=0, \\
                p^{2e-2}+p^{e-1}\sqrt{q}-p^{e-1}\frac{\sqrt{q}+1}{2\ell^{k}};\text{ if }\gamma\in\mathbb{F}_{p}^{*}\text{ and }\delta\neq 0\text{ with }i_{\delta}=\ell^{k}, \\
                p^{2e-2}-p^{e-1}\frac{\sqrt{q}+1}{2\ell^{k}};\text{ if }\gamma\in\mathbb{F}_{p}^{*}\text{ and }\delta\neq 0\text{ with }i_{\delta}\in P\backslash\{\ell^{k}\}, \\
                p^{2e-2};\text{ otherwise}.
            \end{cases}  &
        \end{flalign*}
        \item[\textnormal{(iv)}] When $u\equiv\phi(\ell^{k})\equiv\ell^{k-1}\pmod{p}$, then
        \begin{flalign*}
            N_{\gamma,\delta}^{(1)} &=\begin{cases}
                p^{e-1}\frac{q-1}{2\ell^{k-1}};\text{ if }\gamma\in\mathbb{F}_{p}^{*}\text{ and }\delta=0, \\
                p^{2e-2}+p^{e-1}\sqrt{q}-p^{e-1}\frac{\sqrt{q}+1}{2\ell^{k-1}};\text{ if }\gamma\in\mathbb{F}_{p}^{*}\text{ and }\delta\neq 0\text{ with }i_{\delta}\in\{0\}\cup P_{1}^{(2)}\cup P_{3}^{(2)}, \\
                p^{2e-2}-p^{e-1}\frac{\sqrt{q}+1}{2\ell^{k-1}};\text{ if }\gamma\in\mathbb{F}_{p}^{*}\text{ and }\delta\neq 0\text{ with }i_{\delta}\in P\backslash(\{0\}\cup P_{1}^{(2)}\cup P_{3}^{(2)}), \\
                p^{2e-2};\text{ otherwise}.
            \end{cases}  &
        \end{flalign*}
        \item[\textnormal{(v)}] When $u\equiv-\phi(\ell^{k})\equiv-\ell^{k-1}\pmod{p}$, then
        \begin{flalign*}
            N_{\gamma,\delta}^{(1)} &=\begin{cases}
                p^{e-1}\frac{q-1}{2\ell^{k-1}};\text{ if }\gamma\in\mathbb{F}_{p}^{*}\text{ and }\delta=0, \\
                p^{2e-2}+p^{e-1}\sqrt{q}-p^{e-1}\frac{\sqrt{q}+1}{2\ell^{k-1}};\text{ if }\gamma\in\mathbb{F}_{p}^{*}\text{ and }\delta\neq 0\text{ with }i_{\delta}\in\{\ell^{k}\}\cup P_{1}^{(3)}\cup P_{3}^{(3)}, \\
                p^{2e-2}-p^{e-1}\frac{\sqrt{q}+1}{2\ell^{k-1}};\text{ if }\gamma\in\mathbb{F}_{p}^{*}\text{ and }\delta\neq 0\text{ with }i_{\delta}\in P\backslash(\{\ell^{k}\}\cup P_{1}^{(3)}\cup P_{3}^{(3)}), \\
                p^{2e-2};\text{ otherwise}.
            \end{cases}  &
        \end{flalign*}
         \item[\textnormal{(vi)}] When $u^{2}\not\equiv\phi(\ell^{k})^{2}\pmod{p}$ and $u\equiv\ell^{(k-1)}\pmod{p}$, then
        \begin{flalign*}
            N_{\gamma,\delta}^{(1)} &=\begin{cases}
                p^{e-1}(\ell-1)\frac{q-1}{2\ell^{k}};\text{ if }\gamma\in\mathbb{F}_{p}^{*}\text{ and }\delta=0, \\
                p^{2e-2}+p^{e-1}\sqrt{q}-p^{e-1}(\ell-1)\frac{\sqrt{q}+1}{2\ell^{k}};\text{ if }\gamma\in\mathbb{F}_{p}^{*}\text{ and }\delta\neq 0\text{ with }i_{\delta}\in P_{1}^{(2)}\cup P_{3}^{(2)}, \\
                p^{2e-2}-p^{e-1}(\ell-1)\frac{\sqrt{q}+1}{2\ell^{k}};\text{ if }\gamma\in\mathbb{F}_{p}^{*}\text{ and }\delta\neq 0\text{ with }i_{\delta}\in P\backslash(P_{1}^{(2)}\cup P_{3}^{(2)}), \\
                p^{2e-2};\text{ otherwise}.
            \end{cases}  &
        \end{flalign*}
        \item[\textnormal{(vii)}] When $u^{2}\not\equiv\phi(\ell^{k})^{2}\pmod{p}$ and $u\equiv-\ell^{(k-1)}\pmod{p}$, then
        \begin{flalign*}
            N_{\gamma,\delta}^{(1)} &=\begin{cases}
                p^{e-1}(\ell-1)\frac{q-1}{2\ell^{k}};\text{ if }\gamma\in\mathbb{F}_{p}^{*}\text{ and }\delta=0, \\
                p^{2e-2}+p^{e-1}\sqrt{q}-p^{e-1}(\ell-1)\frac{\sqrt{q}+1}{2\ell^{k}};\text{ if }\gamma\in\mathbb{F}_{p}^{*}\text{ and }\delta\neq 0\text{ with }i_{\delta}\in P_{1}^{(3)}\cup P_{3}^{(3)}, \\
                p^{2e-2}-p^{e-1}(\ell-1)\frac{\sqrt{q}+1}{2\ell^{k}};\text{ if }\gamma\in\mathbb{F}_{p}^{*}\text{ and }\delta\neq 0\text{ with }i_{\delta}\in P\backslash(P_{1}^{(3)}\cup P_{3}^{(3)}), \\
                p^{2e-2};\text{ otherwise}.
            \end{cases}  &
        \end{flalign*}
        \end{enumerate}
       
    \end{enumerate}
    \end{lemma}
   \begin{proof}
   According to the definition, we have
       \begin{align}\label{EQ6}
           N_{\gamma,\delta}^{(1)} &=|\{(x,y)\in\mathbb{F}_{q}^{2}\backslash\{(0,0)\}:\operatorname{Tr}_{1}^{e}(x+y^{\frac{q-1}{2\ell^{k}}})=u\text{ and }\operatorname{Tr}_{1}^{e}(\gamma x+\delta y)=0\}| \nonumber\\
           &= \sum_{x,y\in\mathbb{F}_{q}}\left(\frac{1}{p}\sum_{z\in\mathbb{F}_{p}}\zeta_{p}^{z\{\operatorname{Tr}_{1}^{e}(x+y^{\frac{q-1}{2\ell^{k}}})-u\}}\right)\times\left(\frac{1}{p}\sum_{\delta\in\mathbb{F}_{p}}\zeta_{p}^{\delta\operatorname{Tr}_{1}^{e}(\gamma x+\delta y)}\right)-\frac{1}{p}\sum_{z\in\mathbb{F}_{p}}\zeta_{p}^{-uz} \nonumber \\
           &= \frac{1}{p^{2}}\sum_{x,y\in\mathbb{F}_{q}}\left(\sum_{z\in\mathbb{F}_{p}^{*}}\zeta_{p}^{z\{\operatorname{Tr}_{1}^{e}(x+y^{\frac{q-1}{2\ell^{k}}})-u\}}+1\right)\left(\sum_{\delta\in\mathbb{F}_{p}^{*}}\zeta_{p}^{\delta\operatorname{Tr}_{1}^{e}(\gamma x+\delta y)}+1\right)-\frac{1}{p}\sum_{z\in\mathbb{F}_{p}}\zeta_{p}^{-uz} \nonumber \\
           &= \frac{1}{p^{2}}\left(q^{2}+\sum_{x,y\in\mathbb{F}_{q}}\sum_{z,\delta\in\mathbb{F}_{p}^{*}}\zeta_{p}^{z(\operatorname{Tr}_{1}^{e}(x+y^{\frac{q-1}{2\ell^{k}}})-u)+\delta\operatorname{Tr}_{1}^{e}(\gamma x+\delta y)}\right)-\frac{1}{p}\sum_{z\in\mathbb{F}_{p}}\zeta_{p}^{-uz} \nonumber \\
           &=p^{2e-2}+\frac{1}{p^{2}}\Omega_{1}-\frac{1}{p}\sum_{z\in\mathbb{F}_{p}}\zeta_{p}^{-uz},  
       \end{align}
       where $\Omega_{1}=\sum_{x,y\in\mathbb{F}_{q}}\sum_{z,w\in\mathbb{F}_{p}^{*}}\zeta_{p}^{z(\operatorname{Tr}_{1}^{e}(x+y^{\frac{q-1}{2\ell^{k}}})-u)+w\operatorname{Tr}_{1}^{e}(\gamma x+\delta y)}$.
       \vskip 1pt
       Observe that
       \begin{align}\label{EQ7}
       \hskip 15pt    \Omega_{1} &= \sum_{z\in\mathbb{F}_{p}^{*}}\zeta_{p}^{-uz}\sum_{w\in\mathbb{F}_{p}^{*}}\sum_{x\in\mathbb{F}_{q}}\zeta_{p}^{z\operatorname{Tr}_{1}^{e}\left((1+\gamma\frac{w}{z})x\right)}\times\sum_{y\in\mathbb{F}_{q}}\zeta_{p}^{z\operatorname{Tr}_{1}^{e}(y^{\frac{q-1}{2\ell^{k}}}+\delta\frac{w}{z}y)} \nonumber & \\
           &= \sum_{z\in\mathbb{F}_{p}^{*}}\zeta_{p}^{-uz}\sum_{w\in\mathbb{F}_{p}^{*}}\sum_{x\in\mathbb{F}_{q}}\zeta_{p}^{z\operatorname{Tr}_{1}^{e}\left((1+\gamma w)x\right)}\times\sum_{y\in\mathbb{F}_{q}}\zeta_{p}^{z\operatorname{Tr}_{1}^{e}(y^{\frac{q-1}{2\ell^{k}}}+\delta w y)} &
       \end{align}
       Eq. $(\ref{EQ7})$ holds because for a fixed $z$, as $w$ varies over $\mathbb{F}_{p}^{*}$, so does $\frac{w}{z}$. If $\gamma=0$ and $\delta\neq 0$, it is easy to see that $\Omega_{1}=0$. It then implies that $N_{0,\delta}^{(1)}=\begin{cases}
           p^{2e-2}-1;\text{ if }u=0, \\
           p^{2e-2};\text{ if }u\in\mathbb{F}_{p}^{*}.
       \end{cases}$
       \vskip 1pt
       If $\gamma\neq 0$, then from $(\ref{EQ7})$, we have
       \begin{align*}
           \Omega_{1} &=\begin{cases}
               0;\text{ if }\gamma\not\in\mathbb{F}_{p}^{*}, \\
               q\sum_{z\in\mathbb{F}_{p}^{*}}\sum_{y\in\mathbb{F}_{q}}\zeta_{p}^{z\{\operatorname{Tr}_{1}^{e}(y^{\frac{q-1}{2\ell^{k}}}-\frac{\delta}{\gamma}y)-u\}};\text{ if }\gamma\in\mathbb{F}_{p}^{*}.
           \end{cases}
       \end{align*}
       When $\gamma\in\mathbb{F}_{p}^{*}$, from the definition in Eq. $(\ref{EQ3})$ and Lemma \ref{Le2} (e), we have
       \begin{align}
     \hskip 15pt      \Omega_{1} &= q\sum_{z\in\mathbb{F}_{p}^{*}}\zeta_{p}^{-uz}\left(S_{2\ell^{k}}(z,-\frac{\delta}{\gamma}z)+1\right) \nonumber &\\
           &= q\sum_{z\in\mathbb{F}_{p}^{*}}\zeta_{p}^{-uz}\left(S_{2\ell^{k}}(z,\delta)+1\right)  \nonumber & \\
           &= q\left(w(u,\delta)+\sum_{z\in\mathbb{F}_{p}^{*}}\zeta_{p}^{-uz}\right)   &
       \end{align}
       For the case of $u=0$ and $\ell\equiv 1\pmod{p}$, by using Lemma \ref{Le5}, we obtain
       \begin{align}\label{EQ8}
    \hskip 15pt       \Omega_{1} &=\begin{cases}
        q\left(q(p-1)-p(\ell-1)\frac{q-1}{\ell^{k}}\right);\text{ if }\delta=0, \\
           pq\left((\ell-1)\frac{\sqrt{q}+1}{\ell^{k}}-\sqrt{q}\right);\text{ if }\delta\neq 0\text{ with }i_{\delta}\in P_{1}^{(2)}\cup P_{1}^{(3)}\cup P_{3}^{(2)}\cup P_{3}^{(3)},  \\
           pq(\ell-1)\frac{\sqrt{q}+1}{\ell^{k}};\text{ otherwise}. 
    \end{cases} &
       \end{align}
       This implies that when $u=0$ and $\ell\equiv 1\pmod{p}$, by combining Eq. $(\ref{EQ6})$ and $(\ref{EQ8})$, we obtain
       \begin{flalign*}
        \hspace{15pt}   N_{\gamma,\delta}^{(1)} &=\begin{cases}
               p^{2e-2}+\frac{1}{p^{2}}\Omega_{1}-1;\text{ if }\gamma\in\mathbb{F}_{p}^{*}\text{ and }\delta\in\mathbb{F}_{q}, \\
               p^{2e-2}-1;\text{ otherwise}. 
           \end{cases}  \\
           &=\begin{cases}
                p^{2e-1}-1-p^{e-1}(\ell-1)\frac{q-1}{\ell^{k}};\text{ if }\gamma\in\mathbb{F}_{p}^{*}\text{ and }\delta=0, \\
                    p^{2e-2}-1+p^{e-1}(\ell-1)\frac{\sqrt{q}+1}{\ell^{k}}-p^{e-1}\sqrt{q};\text{ if }\gamma\in\mathbb{F}_{p}^{*}\text{ and }\delta\neq 0\text{ with }\\\hspace{5.4cm}i_{\delta}\in P_{1}^{(2)}\cup P_{1}^{(3)}\cup P_{3}^{(2)}\cup P_{3}^{(3)}, \\
                    p^{2e-2}-1+p^{e-1}(\ell-1)\frac{\sqrt{q}+1}{\ell^{k}};\text{ if }\gamma\in\mathbb{F}_{p}^{*}\text{ and }\delta\neq 0\text{ with }i_{\delta}\in P\backslash(P_{1}^{(2)}\cup P_{1}^{(3)}\cup P_{3}^{(2)}\cup P_{3}^{(3)}), \\
                    p^{2e-2}-1;\text{ otherwise}.
                \end{cases} &
           \end{flalign*}
           Consequently, $N_{\gamma,\delta}^{(1)}$ can be determined accordingly for the remaining cases. The proof is completed.
      \end{proof}
      
We are now ready to prove the main results of this section.      
\vspace{1mm}
      \begin{theorem}\label{Th1}
           Assume that $q=p^{e}$ and $e=\phi(\ell^{k})$. Let $\mathcal{C}_{{D}_{u}}=\{(\operatorname{Tr}_{1}^{e}(\gamma x+\delta y))_{(x,y)\in D_{u}}:\gamma,\delta\in\mathbb{F}_{q}\}$, where $D_{u}$ is defined in $(\ref{EQ9})$ with $u\in\mathbb{F}_{p}$. Then,
          \begin{enumerate}
              \item[\textnormal{(1)}] If $u=0$ and $\ell\equiv 1\pmod{p}$ , then $\mathcal{C}_{{D}_{0}}$ is at most four-weight $[p^{2e-1}-1,2e ]$ linear code over $\mathbb{F}_{p}$ and its weight distribution given in Table $\ref{Table8}$. \\
               \end{enumerate}              
\begin{center}
\captionof{table}{}
\label{Table8}
\begin{tabular}{c c} 
\hline
\textnormal{Weight} & \textnormal{Frequency} \\ 
\hline
 \textnormal{0} & \textnormal{1} \\ 
$p^{e-1}(\ell-1)\frac{q-1}{\ell^{k}}$ & $p-1$ \\ 
$p^{2e-2}(p-1)+p^{e-1}\left(\sqrt{q}-(\ell-1)\frac{\sqrt{q}+1}{\ell^{k}}\right)$ & $(p-1)(\ell-1)\frac{q-1}{\ell^{k}}$ \\ 
$p^{2e-2}(p-1)-p^{e-1}(\ell-1)\frac{\sqrt{q}+1}{\ell^{k}}$ & $(p-1)(\ell^{k}-\ell+1)\frac{q-1}{\ell^{k}}$ \\
$p^{2e-2}(p-1)$ & $(q-1)+(q-p)q$ \\
\hline 
\end{tabular}
\end{center}
\begin{enumerate}
    \item[\textnormal{(2)}] If $u=0$ and $\ell\not\equiv 1\pmod{p}$, then $\mathcal{C}_{{D}_{0}}$ is at most four-weight $[p^{2e-1}-1,2e]$ linear code over $\mathbb{F}_{p}$ and its weight distribution given in Table $\ref{Table9}$. \\
\end{enumerate}
\centering
\captionof{table}{}
\label{Table9}
\begin{tabular}{c c} 
\hline
\textnormal{Weight} & \textnormal{Frequency} \\ 
\hline
 \textnormal{0} & \textnormal{1} \\ 
$p^{e-1}\frac{q-1}{\ell^{k-1}}$ & $p-1$ \\ 
$p^{2e-2}(p-1)+p^{e-1}\left(\sqrt{q}-\frac{\sqrt{q}+1}{\ell^{k-1}}\right)$ & $(p-1)\frac{q-1}{\ell^{k-1}}$ \\ 
$p^{2e-2}(p-1)-p^{e-1}\frac{\sqrt{q}+1}{\ell^{k-1}}$ & $(p-1)(\ell^{k-1}-1)\frac{q-1}{\ell^{k-1}}$ \\
$p^{2e-2}(p-1)$ & $(q-1)+(q-p)q$ \\
\hline
\end{tabular}
\begin{enumerate}
    \item[\textnormal{(3)}] If $u\neq 0$, $u^{2}\not\equiv \phi(\ell^{k})^{2}\pmod{p}$ and $u^{2}\not\equiv\ell^{2(k-1)}\pmod{p}$, then $\mathcal{C}_{{D}_{u}}$ is a two-weight $[p^{2e-1},2e]$ linear code over $\mathbb{F}_{p}$ and its weight distribution given in Table $\ref{Table10}$. \\
\end{enumerate}
\centering
\captionof{table}{}
\label{Table10}
\begin{tabular}{c c} 
\hline
\textnormal{Weight} & \textnormal{Frequency} \\ 
\hline
 \textnormal{0} & \textnormal{1} \\ 
$p^{2e-1}$ & $p-1$ \\ 
$p^{2e-2}(p-1)$ & $q^{2}-p$ \\ 
\hline
\end{tabular}
\begin{enumerate}
    \item[\textnormal{(4)}] If $u\neq 0$, $u^{2}\equiv \phi(\ell^{k})^{2}\pmod{p}$ and $u^{2}\not\equiv\ell^{2(k-1)}\pmod{p}$, then $\mathcal{C}_{{D}_{u}}$ is at most four-weight $[p^{2e-1},2e]$ linear code over $\mathbb{F}_{p}$ and its weight distribution given in Table $\ref{Table11}$. \\
    \end{enumerate}
\centering
\captionof{table}{}
\label{Table11}
\begin{tabular}{c c} 
\hline
\textnormal{Weight} & \textnormal{Frequency} \\ 
\hline
 \textnormal{0} & \textnormal{1} \\ 
$p^{2e-1}-p^{e-1}\frac{q-1}{2\ell^{k}}$ & $p-1$ \\ 
$p^{2e-2}(p-1)-p^{e-1}\sqrt{q}+p^{e-1}\frac{\sqrt{q}+1}{2\ell^{k}}$ & $(p-1)\frac{q-1}{2\ell^{k}}$ \\ 

$p^{2e-2}(p-1)+p^{e-1}\frac{\sqrt{q}+1}{2\ell^{k}}$ & $(p-1)(2\ell^{k}-1)\frac{q-1}{2\ell^{k}}$ \\

$p^{2e-2}(p-1)$ & $(q-1)+q(q-p)$ \\
\hline
\end{tabular}
\begin{enumerate}
    \item[\textnormal{(5)}] If $u\neq 0$ with either $u\equiv \phi(\ell^{k})\equiv\ell^{k-1}\pmod{p}$ or $u\equiv-\phi(\ell^{k})\equiv-\ell^{k-1}\pmod{p}$, then $\mathcal{C}_{{D}_{u}}$ is at most four-weight $[p^{2e-1},2e]$ linear code over $\mathbb{F}_{p}$ and its weight distribution given in Table $\ref{Table12}$. \\
    \end{enumerate}
\centering
\captionof{table}{}
\label{Table12}
\begin{tabular}{c c} 
\hline
\textnormal{Weight} & \textnormal{Frequency} \\ 
\hline
 \textnormal{0} & \textnormal{1} \\ 
$p^{2e-1}-p^{e-1}\frac{q-1}{2\ell^{k-1}}$ & $p-1$ \\ 
$p^{2e-2}(p-1)-p^{e-1}\sqrt{q}+p^{e-1}\frac{\sqrt{q}+1}{2\ell^{k-1}}$ & $(p-1)\frac{q-1}{2\ell^{k-1}}$ \\ 

$p^{2e-2}(p-1)+p^{e-1}\frac{\sqrt{q}+1}{2\ell^{k-1}}$ & $(p-1)(2\ell^{k-1}-1)\frac{q-1}{2\ell^{k-1}}$ \\

$p^{2e-2}(p-1)$ & $(q-1)+q(q-p)$ \\
\hline
\end{tabular}
\begin{enumerate}
    \item[\textnormal{(6)}] If $u\neq 0$, $u^{2}\not\equiv\phi(\ell^{k})^{2}\pmod{p}$ and $u^{2}\equiv\ell^{2(k-1)}\pmod{p}$, then $\mathcal{C}_{{D}_{u}}$ is at most four-weight $[p^{2e-1},2e]$ linear code over $\mathbb{F}_{p}$ and its weight distribution given in Table $\ref{Table13}$. \\
    \end{enumerate}
\centering
\captionof{table}{}
\label{Table13}
\begin{tabular}{c c} 
\hline
\textnormal{Weight} & \textnormal{Frequency} \\ 
\hline
 \textnormal{0} & \textnormal{1} \\ 
$p^{2e-1}-p^{e-1}(\ell-1)\frac{q-1}{2\ell^{k}}$ & $p-1$ \\ 
$p^{2e-2}(p-1)-p^{e-1}\sqrt{q}+p^{e-1}(\ell-1)\frac{\sqrt{q}+1}{2\ell^{k}}$ & $(p-1)(\ell-1)\frac{q-1}{2\ell^{k}}$ \\ 

$p^{2e-2}(p-1)+p^{e-1}(\ell-1)\frac{\sqrt{q}+1}{2\ell^{k}}$ & $(p-1)(2\ell^{k}-\ell+1)\frac{q-1}{2\ell^{k}}$ \\

$p^{2e-2}(p-1)$ & $(q-1)+q(q-p)$ \\
\hline
\end{tabular}
\end{theorem}             
          \begin{proof}
             We will prove only for the case $u=0$ and $\ell\equiv 1\pmod{p}$ as the remaining cases can be proved in a similar manner.  From Lemma \ref{EQ12} and Eq. $(\ref{EQ11})$, we conclude that $\operatorname{wt}(c(\gamma,\delta))\in\{w_{1},w_{2},w_{3},w_{4}\}$ for every $(\gamma,\delta)\in\mathbb{F}_{q}^{2}\backslash\{(0,0)\}$, where
              \begin{align*}
                  w_{1} &= p^{e-1}(\ell-1)\frac{q-1}{\ell^{k}}, \\
                  w_{2} &= p^{2e-2}(p-1)+p^{e-1}\left(\sqrt{q}-(\ell-1)\frac{\sqrt{q}+1}{\ell^{k}}\right), \\
                  w_{3} &= p^{2e-2}(p-1)-p^{e-1}(\ell-1)\frac{\sqrt{q}+1}{\ell^{k}}, \\
                  w_{4} &= p^{2e-2}(p-1). \\
              \end{align*}
              In the following, we will show their frequencies $A_{w_{1}}$, $A_{w_{2}}$, $A_{w_{3}}$ and $A_{w_{4}}$, respectively. 
              \vskip 1pt
              Recall that, for $\delta\in\mathbb{F}_{q}^{*}$, $\operatorname{Ind}_{\alpha}(\delta)$ be the index of $\delta$ with respect to the base $\alpha$, where $\alpha$ generates $\mathbb{F}_{q}^{*}$ and $i_{\delta}=-\operatorname{Ind}_{\alpha}(\delta)\pmod{2\ell^{k}}$. Since $i_{\delta}\in P=\{0,1,\cdots,2\ell^{k}-1\}$, then for any subset $Q$ of $P$, it is easy to see that $|\{\delta\in\mathbb{F}_{q}^{*}:i_{\delta}\in Q\}|=\frac{q-1}{2\ell^{k}}|Q|$. Note that $|P_{1}^{(2)}|=|P_{1}^{(3)}|=|P_{3}^{(2)}|=|P_{3}^{(3)}|=\frac{\ell-1}{2}$. Therefore, we obtain
              
              \begin{align*}
                  A_{w_{1}} &=|\{(\gamma,\delta)\in\mathbb{F}_{q}^{2}:\gamma\in\mathbb{F}_{p}^{*}\text{ and }\delta=0\}|=(p-1), \\
                  A_{w_{2}} &=|\{(\gamma,\delta)\in\mathbb{F}_{q}^{2}:\gamma\in\mathbb{F}_{p}^{*}\text{ and }\delta\neq 0\text{ with }i_{\delta}\in P_{1}^{(2)}\cup P_{1}^{(3)}\cup P_{3}^{(2)}\cup P_{3}^{(3)}\}|  \\
                  &= (p-1)(\ell-1)\frac{q-1}{\ell^{k}}, \\
                  A_{w_{3}} &=|\{(\gamma,\delta)\in\mathbb{F}_{q}^{2}:\gamma\in\mathbb{F}_{p}^{*}\text{ and }\delta\neq 0\text{ with }i_{\delta}\in P\backslash (P_{1}^{(2)}\cup P_{1}^{(3)}\cup P_{3}^{(2)}\cup P_{3}^{(3)})\}|  \\
                  &= (p-1)(\ell^{k}-\ell+1)\frac{q-1}{\ell^{k}}, \\
                  A_{w_{4}} &=|\{(\gamma,\delta)\in\mathbb{F}_{q}^{2}:\gamma=0\text{ and }\delta\in\mathbb{F}_{q}^{*}\}|+|\{(\gamma,\delta)\in\mathbb{F}_{q}^{2}:\gamma\in\mathbb{F}_{q}\backslash\mathbb{F}_{p}\text{ and }\delta\in\mathbb{F}_{q}\}| \\
                  &= (q-1)+(q-p)q.
              \end{align*}
              From the above arguments, we see that $\operatorname{wt}(c(\gamma,\delta))\neq 0$ for all $(\gamma,\delta)\in\mathbb{F}_{q}^{2}\backslash\{(0,0)\}$. Therefore, $|\mathcal{C}_{D_{0}}|=q^{2}$ and then $\mathcal{C}_{D_{0}}$ is of dimension $2e$. The length of the code $\mathcal{C}_{D_{u}}$ is obvious from Lemma $\ref{Le8}$.  This completes the proof.
          \end{proof}

\begin{remark}\label{Re1}
   If  $u\neq 0$, $u^{2}\equiv \phi(\ell^{k})^{2}\pmod{p}$, and $u^{2}\not\equiv\ell^{2(k-1)}\pmod{p}$, we have $w_{\textit{min}}=p^{2e-2}(p-1)-p^{e-1}\sqrt{q}+p^{e-1}\frac{\sqrt{q}+1}{2\ell^{k}}$ and $w_{\textit{max}}=p^{2e-1}-p^{e-1}\frac{q-1}{2\ell^{k}}$ for $p^{e-1}\geq\frac{q+\sqrt{q}}{2\ell^{k}}$. One can check that $\frac{w_{\textit{min}}}{w_{\textit{max}}}>\frac{p-1}{p}$ when $(p-1)q+1>p\sqrt{q}(2\ell^{k}-1)$. Hence, under these conditions, the $p$-ary linear code $\mathcal{C}_{D_{u}}$ in Theorem $\ref{Th1}(4)$ is minimal.
   \vskip 1 pt
   If $u\neq 0$ with either $u\equiv \phi(\ell^{k})\equiv\ell^{k-1}\pmod{p}$ or $u\equiv-\phi(\ell^{k})\equiv-\ell^{k-1}\pmod{p}$, we have $w_{\textit{min}}=p^{2e-2}(p-1)-p^{e-1}\sqrt{q}+p^{e-1}\frac{\sqrt{q}+1}{2\ell^{k-1}}$ and $w_{\textit{max}}=p^{2e-1}-p^{e-1}\frac{q-1}{2\ell^{k-1}}$ for $k>1$ and $p^{e-1}\geq\frac{q+\sqrt{q}}{2\ell^{k-1}}$. One can check that $\frac{w_{\textit{min}}}{w_{\textit{max}}}>\frac{p-1}{p}$ when $(p-1)q+1>p\sqrt{q}(2\ell^{k-1}-1)$. Hence, under these conditions, the $p$-ary linear code $\mathcal{C}_{D_{u}}$ in Theorem $\ref{Th1}(5)$ is minimal.
 \vskip 1 pt
   If $u\neq 0$, $u^{2}\not\equiv\phi(\ell^{k})^{2}\pmod{p}$, and $u^{2}\equiv\ell^{2(k-1)}\pmod{p}$, we have $w_{\textit{min}}=p^{2e-2}(p-1)-p^{e-1}\sqrt{q}+p^{e-1}(\ell-1)\frac{\sqrt{q}+1}{2\ell^{k}}$ and $w_{\textit{max}}=p^{2e-1}-p^{e-1}(\ell-1)\frac{q-1}{2\ell^{k}}$ for $p^{e-1}\geq(\ell-1)\frac{q+\sqrt{q}}{2\ell^{k}}$. One can check that $\frac{w_{\textit{min}}}{w_{\textit{max}}}>\frac{p-1}{p}$ when $(p-1)q+1>p\sqrt{q}\left(\frac{2\ell^{k}-\ell+1}{\ell-1}\right)$. Hence, under these conditions, the $p$-ary linear code $\mathcal{C}_{D_{u}}$ in Theorem $\ref{Th1}(6)$ is minimal.
\end{remark}
\vspace{1em}
\begin{example}
    Consider $p=5$, $\ell=3$, $k=1$, and $u=0$; then $\ell\not\equiv 1\pmod{p}$ and $e=\phi(6)=2$. The code $\mathcal{C}_{D_{0}}$ has parameters $[124,4,95]$, and its weight enumerator is $1+96z^{95}+524z^{100}+4z^{120}$. By a magma program, the obtained code $\mathcal{C}_{D_{0}}$ is consistent with Theorem $\ref{Th1}(2)$ (see Table $\ref{Table9}$).
\end{example}
\vspace{1em}
\begin{example}
    Consider $p=5$, $\ell=3$, $k=2$, and $u(\neq 0)\in\mathbb{F}_{5}$ such that $u^{2}\equiv (\phi(\ell^k))^2\equiv 1\pmod{5}$ and $u^{2}\not\equiv \ell^{2(k-1)}\equiv 4\pmod{5}$. Then $u\in\{1,4\}$. The linear code $\mathcal{C}_{D_{u}}$ has parameters $[48828125,12,38693750]$, and its weight enumerator is $1+3472z^{38693750}+244078124z^{39062500}+59024z^{39084375}+4z^{46115625}$. Due to the large computation required, using a magma program to verify whether the code $\mathcal{C}_{D_{u}}$ is consistent with Theorem $\ref{Th1}(4)$ is not possible.
\end{example}
\vspace{1em}
\begin{example}
    Consider $p=7$, $\ell=5$, $k=1$, and $u(\neq 0)\in\mathbb{F}_{7}$ such that $u^{2}\equiv(\phi(\ell^k))^2\equiv 2\pmod{7}$ and $u^{2}\not\equiv \ell^{2(k-1)}\equiv 1\pmod{7}$. Then $u\in\{3,4\}$. The linear code $\mathcal{C}_{D_{u}}$ has parameters $[823543,8,690802]$, and its weight enumerator is  $1+1440z^{690802}+5750394z^{705894}+12960z^{707609}+6z^{741223}$.  Due to the large computation required, using a magma program to verify whether the code $\mathcal{C}_{D_{u}}$ is consistent with Theorem $\ref{Th1}(4)$ is not possible.
 \end{example}
\vspace{1em}
\begin{example}
   Consider $p=5$, $\ell=3$, $k=1$, and $u(\neq 0)\in\mathbb{F}_{5}$ such that $u^{2}\equiv\ell^{2(k-1)}\equiv 1\pmod{5}$ and $u^{2}\not\equiv(\phi(\ell^k))^2\equiv 4\pmod{5}$, then $u\in\{1,4\}$ and $e=\phi(6)=2$. The code $\mathcal{C}_{D_{u}}$ has parameters $[125,4,85]$ and its weight enumerator is $1+36z^{85}+524z^{100}+64z^{110}$. By a magma program, the obtained code $\mathcal{C}_{D_{u}}$ is consistent with Theorem $\ref{Th1}(6)$.
\end{example}
\vspace{1em}
The following corollary provides an optimal class of linear codes under a special condition and can be proved directly from Theorem \ref{Th1}.
\vspace{2mm}
\begin{corollary}\label{Cor1}
Let $u\in\mathbb{F}_{p}^{*}$ be such that $u^{2}\not\equiv \phi(\ell^{k})^{2}\pmod{p}$ and $u^{2}\not\equiv\ell^{2(k-1)}\pmod{p}$, then the code $\mathcal{C}_{{D}_{u}}$ defined in Theorem $\ref{Th1}$ is an optimal two-weight $[p^{2e-1},2e,p^{2e-2}(p-1)]$ linear code over $\mathbb{F}_{p}$ that meets the Griesmer bound, and its weight enumerator is $1+(p^{2e}-p)z^{p^{2e-2}(p-1)}+(p-1)z^{p^{2e-1}}$. 
\end{corollary}
\vspace{1em}
\begin{example}
    Consider $p=7$, $\ell=5$, $k=1$, and $u(\neq 0)\in\mathbb{F}_{7}$ such that $u^{2}\not\equiv \phi(\ell^{k})^{2}\equiv 2\pmod{7}$ and $u^{2}\not\equiv \ell^{2(k-1)}\equiv 1\pmod{7}$. Then $u\in\{2,5\}$. The code $\mathcal{C}_{D_{u}}$ has parameters $[823543,8,705894]$ and its weight enumerator is $1+5764794z^{705894}+6z^{823543}$. According to Corollary $\ref{Cor1}$, the code $\mathcal{C}_{D_{u}}$ is optimal with respect to the Griesmer bound.  Due to the large computation required, using a magma program to verify whether the code $\mathcal{C}_{D_{u}}$ is consistent with Theorem $\ref{Th1}(3)$ is not possible. 
\end{example}
          \section{Linear codes from the defining set $D^{'}$}\label{sec4}
          In this section, we always assume $f\in\mathcal{RF}$ such that $f^{*}$ is a quadratic form (as discussed in section \ref{S1}) and $\epsilon_{f}$ is the sign of the Walsh transform of $f(x)$. First, we need the following lemmas for determining the lengths and weight distributions of the linear code $\mathcal{C}_{D^{'}}$ defined in $(\ref{EQ2})$ and $(\ref{EQ10})$.
          \vspace{1mm}
              \begin{lemma}\label{LX1}
                  If $D^{'}=\{(x,y)\in\mathbb{F}_{q}^{2}\backslash(0,0):f(x)+\operatorname{Tr}_{1}^{e}(y^{\frac{q-1}{2\ell^{k}}})=0\}$, where $u\in\mathbb{F}_{p}$. Let $n_{2}=|D^{'}|$. Then 
                  \begin{align*}
                      n_{2}&= \begin{cases}
                          \frac{q^{2}}{p}-1+\frac{\epsilon_{f}\sqrt{p^{*}}^{e}}{p}(q(p-1)-p(\ell-1)\frac{q-1}{\ell^{k}});\text{ if }\ell\equiv 1\pmod{p}, \\
                         \frac{q^{2}}{p}-1+\frac{\epsilon_{f}\sqrt{p^{*}}^{e}}{p}(q(p-1)-p\frac{q-1}{\ell^{k-1}});\text{ if }\ell\not\equiv 1\pmod{p}.
                      \end{cases}
                  \end{align*}
                  \begin{proof}
                     Since $e=\phi(\ell^{k})$ is even, we can say $\sigma_{z}(\sqrt{p^{*}}^{e})=\sqrt{p^{*}}^{e}$. By using the definition of $(\ref{EQ22})$, we have
                       \begin{align}\label{EQ13}
                          n_{2} &= \frac{1}{p}\sum_{x,y\in\mathbb{F}_{q}}\sum_{z\in\mathbb{F}_{p}}\zeta_{p}^{z(f(x)+\operatorname{Tr}_{1}^{e}(y^{\frac{q-1}{2\ell^{k}}}))}-1 \nonumber \\
                          &=\frac{q^{2}}{p}-1+\frac{1}{p}\sum_{x,y\in\mathbb{F}_{q}}\sum_{z\in\mathbb{F}_{p}^{*}}\zeta_{p}^{z(f(x)+\operatorname{Tr}_{1}^{e}(y^{\frac{q-1}{2\ell^{k}}}))} \nonumber \\
                          &=\frac{q^{2}}{p}-1+\frac{1}{p}\sum_{z\in\mathbb{F}_{p}^{*}}\sigma_{z}\left(\epsilon_{f}\sqrt{p^{*}}^{e}\right)(S_{2\ell^{k}}(z,0)+1) \nonumber \\
                          &= \frac{q^{2}}{p}-1+\frac{\epsilon_{f}\sqrt{p^{*}}^{e}}{p}(w(0,0)+p-1)
                      \end{align}
                      Hence, the result follows from Lemma \ref{Le5} (a) and Eq. $(\ref{EQ13})$.               
                  \end{proof}
              \end{lemma}
  Define a notation 
          \begin{equation}\label{EQ14}
              N_{\gamma,\delta}^{(2)}=|\{(x,y)\in D^{'}:\operatorname{Tr}_{1}^{e}(\gamma x+\delta y)=0\}|,
          \end{equation}
       and the Hamming weight of the codeword $c(\gamma,\delta)\in\mathcal{C}_{D^{'}}$ is as follows: 
       \begin{equation}\label{EQ15}
           \operatorname{wt}(c(\gamma,\delta))= n_{2}-N_{\gamma,\delta}^{(2)}.
       \end{equation}

\begin{lemma}\label{LX2}
    Let $\gamma,\delta\in\mathbb{F}_{q}$ with $(\gamma,\delta)\neq (0,0)$ and $N_{\gamma,\delta}^{(2)}$ be defined in $(\ref{EQ14})$. For $b\in\mathbb{F}_{q}^{*}$, let $\operatorname{Ind}_{\alpha}(b)$ be the index of $b$ with respect to the base $\alpha$, and $i_{b}=-\operatorname{Ind}_{\alpha}(b)\pmod{2\ell^{k}}$. Assume that the set $P$ and its partitions be defined as in section $\ref{ssec2.1}$. Then the explicit value of $N_{\gamma,\delta}^{(2)}$ is given as follows:
    \vspace{1mm}
    \begin{enumerate}
        \item[\textnormal{1)}] When $f^{*}(\gamma)=0$ $(\gamma\neq 0)$ and $\delta=0$, we have
        \begin{align*}
            N_{\gamma,\delta}^{(2)}&=\begin{cases}
                \frac{q^{2}}{p^{2}}+\epsilon_{f}\sqrt{p^{*}}^{e}\left(\frac{q}{p}(p-1)-\frac{(q-1)(\ell-1)}{\ell^{k}}\right)-1;\text{ if }\ell\equiv 1\pmod{p}, \\
                \frac{q^{2}}{p^{2}}+\epsilon_{f}\sqrt{p^{*}}^{e}\left(\frac{q}{p}(p-1)-\frac{(q-1)}{\ell^{k-1}}\right)-1;\text{ if }\ell\not\equiv 1\pmod{p}.
            \end{cases}
        \end{align*}
        \item[\textnormal{2)}] When $f^{*}(\gamma)=0$ and $\delta\neq 0$, we have
        \begin{align*}
            N_{\gamma,\delta}^{(2)}&=\begin{cases}
                \frac{q^{2}}{p^{2}}+\epsilon_{f}\sqrt{p^{*}}^{e}\left(\frac{\sqrt{q}(\sqrt{q}-p)(p-1)}{p^{2}}-\frac{(\ell-1)(\sqrt{q}+1)(\sqrt{q}-p)}{p\ell^{k}}\right)-1;\text{ if }i_{\delta}\in P_{1}^{(2)}\cup P_{1}^{(3)}\cup P_{3}^{(2)}\cup P_{3}^{(3)}, \\
                \frac{q^{2}}{p^{2}}+\epsilon_{f}\sqrt{p^{*}}^{e}\left(\frac{q(p-1)}{p^{2}}-\frac{(\ell-1)(\sqrt{q}+1)(\sqrt{q}-p)}{p\ell^{k}}\right)-1;\text{ otherwise,}
            \end{cases}
        \end{align*}
        if $\ell\equiv 1\pmod{p}$; and 
        \begin{align*}
            N_{\gamma,\delta}^{(2)} &=\begin{cases}
                \frac{q^{2}}{p^{2}}+\epsilon_{f}\sqrt{p^{*}}^{e}\left(\frac{\sqrt{q}(p-1)(\sqrt{q}-p)}{p^{2}}-\frac{(\sqrt{q}+1)(\sqrt{q}-p)}{p\ell^{k-1}}\right)-1;\text{ if }i_{\delta}\in \{0,\ell^{k}\}\cup P_{1}^{(2)}\cup P_{1}^{(3)}\cup\\\hspace{10.4cm} P_{3}^{(2)}\cup P_{3}^{(3)}, \\
                \frac{q^{2}}{p^{2}}+\epsilon_{f}\sqrt{p^{*}}^{e}\left(\frac{q(p-1)}{p^{2}}-\frac{(\sqrt{q}+1)(\sqrt{q}-p)}{p\ell^{k-1}}\right)-1;\text{ otherwise,} 
            \end{cases}
        \end{align*}
        if $\ell\not\equiv 1\pmod{p}$. \\
        \item[\textnormal{3)}] When $f^{*}(\gamma)\neq 0$ and $\delta=0$, we have
        \begin{align*}
            N_{\gamma,\delta}^{(2)} &=\begin{cases}
                \frac{q^{2}}{p^{2}}+\epsilon_{f}\sqrt{p^{*}}^{e}\eta_{1}(f^{*}(a))\frac{(\ell-1)(q-1)}{p\ell^{k}}-1;\text{ if }\ell\equiv 1\pmod{p}\text{ and }p\equiv 1\pmod{4}, \\
                \frac{q^{2}}{p^{2}}+\epsilon_{f}\sqrt{p^{*}}^{e}\eta_{1}(f^{*}(a))(\eta_{1}(t_{1})+(\ell-1)\eta_{1}(t_{2}))\frac{q-1}{p\ell^{k}}-1;\text{ if }\ell\not\equiv 1\pmod{p}\\\hspace{7.7cm}\text{ and }p\equiv 1\pmod{4}, \\
                \frac{q^{2}}{p^{2}}-1;\text{ if }p\equiv 3\pmod{4},
            \end{cases}
        \end{align*}
        where $t_{1}=\phi(\ell^{k})\pmod{p}$ and $t_{2}=\ell^{k-1}\pmod{p}$. \\
        \item[\textnormal{4)}] When $f^{*}(\gamma)\neq 0$ and $\delta\neq 0$, we have
        \begin{align*}
            N_{\gamma,\delta}^{(2)} &=\begin{cases}
                \frac{q^{2}}{p^{2}}+\epsilon_{f}\sqrt{p^{*}}^{e}\left(\frac{q(p-1)}{p^{2}}+\frac{\sqrt{q}}{p}(1+\eta_{1}(f^{*}(a)))-\frac{(\ell-1)(\sqrt{q}+1)}{p\ell^{k}}(\sqrt{q}+\eta_{1}(f^{*}(a)))\right)-1;\\
               \hspace{7.8cm} \text{ if }i_{\delta}\in P_{1}^{(2)}\cup P_{1}^{(3)}\cup P_{3}^{(2)}\cup P_{3}^{(3)}, \\
                \frac{q^{2}}{p^{2}}+\epsilon_{f}\sqrt{p^{*}}^{e}\left(\frac{q(p-1)}{p^{2}}-\frac{(\ell-1)(\sqrt{q}+1)}{p\ell^{k}}(\sqrt{q}+\eta_{1}(f^{*}(a)))\right)-1;\text{ otherwise,}
            \end{cases}
        \end{align*}
        if $\ell\equiv 1\pmod{p}$ and $p\equiv 1\pmod{4}$,
        \begin{align*}
            N_{\gamma,\delta}^{(2)} &=\begin{cases}
                \frac{q^{2}}{p^{2}}+\epsilon_{f}\sqrt{p^{*}}^{e}\times\\ \left(\frac{q(p-1)}{p^{2}}+\frac{\sqrt{q}}{p}(1+\eta_{1}(f^{*}(\gamma))\eta_{1}(t_{1}))-\frac{\sqrt{q}+1}{p\ell^{k}}(\ell\sqrt{q}+\eta_{1}(f^{*}(\gamma))(\eta_{1}(t_{1})+(\ell-1)\eta_{1}(t_{2})))\right)-1;\\ \hspace{11.3cm}\text{ if }i_{\delta}\in\{0,\ell^{k}\}, \\
                \frac{q^{2}}{p^{2}}+\epsilon_{f}\sqrt{p^{*}}^{e}\times\\ \left(\frac{q(p-1)}{p^{2}}+\frac{\sqrt{q}}{p}(1+\eta_{1}(f^{*}(\gamma))\eta_{1}(t_{2}))-\frac{\sqrt{q}+1}{p\ell^{k}}(\ell\sqrt{q}+\eta_{1}(f^{*}(\gamma))(\eta_{1}(t_{1})+(\ell-1)\eta_{1}(t_{2})))\right)-1;\\ \hspace{8.4cm}\text{ if }i_{\delta}\in P_{1}^{(2)}\cup P_{1}^{(3)}\cup P_{3}^{(2)}\cup P_{3}^{(3)}, \\
                \frac{q^{2}}{p^{2}}+\epsilon_{f}\sqrt{p^{*}}^{e} \left(\frac{q(p-1)}{p^{2}}-\frac{\sqrt{q}+1}{p\ell^{k}}(\ell\sqrt{q}+\eta_{1}(f^{*}(\gamma))(\eta_{1}(t_{1})+(\ell-1)\eta_{1}(t_{2})))\right)-1;\text{ otherwise,}
            \end{cases}
        \end{align*}
        if $\ell\not\equiv 1\pmod{p}$ and $p\equiv 1\pmod{4}$,
        \begin{align*}
            N_{\gamma,\delta}^{(2)} &=\begin{cases}
            \frac{q^{2}}{p^2}+\epsilon_{f}\sqrt{p^{*}}^{e}\left(\frac{q(p-1)}{p^{2}}+\frac{\sqrt{q}}{p}(1+\eta_{1}(f^{*}(\gamma)))-\frac{(\ell-1)(q+\sqrt{q})}{p\ell^{k}}\right)-1;\text{ if }i_{\delta}\in P_{1}^{(3)}\cup P_{3}^{(3)}, \\
             \frac{q^{2}}{p^2}+\epsilon_{f}\sqrt{p^{*}}^{e}\left(\frac{q(p-1)}{p^{2}}+\frac{\sqrt{q}}{p}(1-\eta_{1}(f^{*}(\gamma)))-\frac{(\ell-1)(q+\sqrt{q})}{p\ell^{k}}\right)-1;\text{ if }i_{\delta}\in P_{1}^{(2)}\cup P_{3}^{(2)}, \\
              \frac{q^{2}}{p^2}+\epsilon_{f}\sqrt{p^{*}}^{e}\left(\frac{q(p-1)}{p^{2}}-\frac{(\ell-1)(q+\sqrt{q})}{p\ell^{k}}\right)-1;\text{ otherwise, }
            \end{cases}
        \end{align*}
        if $\ell\equiv 1\pmod{p}$ and $p\equiv 3\pmod{4}$,
        \begin{align*}
            N_{\gamma,\delta}^{(2)} &=\begin{cases}
                \frac{q^{2}}{p^{2}}+\epsilon_{f}\sqrt{p^{*}}^{e}\left(\frac{q(p-1)}{p^{2}}+\frac{\sqrt{q}}{p}(1-\eta_{1}(f^{*}(\gamma))\eta_{1}(t_{1}))-\frac{(q+\sqrt{q})}{p\ell^{k-1}}\right)-1;\text{ if }i_{\delta}=0, \\
                \frac{q^{2}}{p^{2}}+\epsilon_{f}\sqrt{p^{*}}^{e}\left(\frac{q(p-1)}{p^{2}}+\frac{\sqrt{q}}{p}(1+\eta_{1}(f^{*}(\gamma))\eta_{1}(t_{1}))-\frac{(q+\sqrt{q})}{p\ell^{k-1}}\right)-1;\text{ if }i_{\delta}=\ell^{k}, \\
                 \frac{q^{2}}{p^{2}}+\epsilon_{f}\sqrt{p^{*}}^{e}\left(\frac{q(p-1)}{p^{2}}+\frac{\sqrt{q}}{p}(1+\eta_{1}(f^{*}(\gamma))\eta_{1}(t_{2}))-\frac{(q+\sqrt{q})}{p\ell^{k-1}}\right)-1;\text{ if }i_{\delta}\in P_{1}^{(3)}\cup P_{3}^{(3)}, \\
          \frac{q^{2}}{p^{2}}+\epsilon_{f}\sqrt{p^{*}}^{e}\left(\frac{q(p-1)}{p^{2}}+\frac{\sqrt{q}}{p}(1-\eta_{1}(f^{*}(\gamma))\eta_{1}(t_{2}))-\frac{(q+\sqrt{q})}{p\ell^{k-1}}\right)-1;\text{ if }i_{\delta}\in P_{1}^{(2)}\cup P_{3}^{(2)}, \\
        \frac{q^{2}}{p^{2}}+\epsilon_{f}\sqrt{p^{*}}^{e}\left(\frac{q(p-1)}{p^{2}}-\frac{(q+\sqrt{q})}{p\ell^{k-1}}\right)-1;\text{ otherwise, }
            \end{cases}
        \end{align*}
        if $\ell\not\equiv 1\pmod{p}$ and $p\equiv 3\pmod{4}$, 
        \vskip 1pt 
        where $t_{1}=\phi(\ell^{k})\pmod{p}$ and $t_{2}=\ell^{k-1}\pmod{p}$. 
    \end{enumerate}
    \begin{proof}
        By the definition in $(\ref{EQ3})$ and using the fact $S_{2\ell^{k}}(z,w \delta)=S_{2\ell^{k}}(z,\delta)$ for any $\delta\in\mathbb{F}_{q}$ and $z,w\in\mathbb{F}_{p}^{*}$, we obtain  
        \begin{align}\label{EQ17}
            N_{\gamma,\delta}^{(2)} &=|\{(x,y)\in\mathbb{F}_{q}^{2}\backslash\{(0,0)\}:f(x)+\operatorname{Tr}_{1}^{e}(y^{\frac{q-1}{2\ell^{k}}})=0\text{ and }\operatorname{Tr}_{1}^{e}(\gamma x+\delta y)=0\}| \nonumber \\
            &= \sum_{x,y\in\mathbb{F}_{q}}\left(\frac{1}{p}\sum_{z\in\mathbb{F}_{p}}\zeta_{p}^{z\{f(x)+\operatorname{Tr}_{1}^{e}(y^{\frac{q-1}{2\ell^{k}}})\}}\right)\times\left(\frac{1}{p}\sum_{w\in\mathbb{F}_{p}}\zeta_{p}^{w\operatorname{Tr}_{1}^{e}(\gamma x+\delta y)}\right)-1 \nonumber \\
            &=\frac{1}{p^{2}}\sum_{x,y\in\mathbb{F}_{q}}\left(\sum_{z\in\mathbb{F}_{p}^{*}}\zeta_{p}^{z\{f(x)+\operatorname{Tr}_{1}^{e}(y^{\frac{q-1}{2\ell^{k}}})\}}+1\right)\left(\sum_{w\in\mathbb{F}_{p}^{*}}\zeta_{p}^{w\operatorname{Tr}_{1}^{e}(\gamma x+\delta y)}+1\right)-1 \nonumber \\
            &= \frac{1}{p^{2}}\left(q^{2}+\sum_{x,y\in\mathbb{F}_{q}}\sum_{z\in\mathbb{F}_{p}^{*}}\zeta_{p}^{z\{f(x)+\operatorname{Tr}_{1}^{e}(y^{\frac{q-1}{2\ell^{k}}})\}}+\sum_{x,y\in\mathbb{F}_{q}}\sum_{z,w\in\mathbb{F}_{p}^{*}}\zeta_{p}^{z\{f(x)+\operatorname{Tr}_{1}^{e}(y^{\frac{q-1}{2\ell^{k}}})\}+w\operatorname{Tr}_{1}^{e}(\gamma x+\delta y)}\right)-1 \nonumber \\
            &= \frac{1}{p^{2}}\left(q^{2}+\epsilon_{f}\sqrt{p^{*}}^{e}(w(0,0)+p-1)+\sum_{z,w\in\mathbb{F}_{p}^{*}}\sum_{x\in\mathbb{F}_{q}}\zeta_{p}^{z\{f(x)+\operatorname{Tr}_{1}^{e}(\gamma\frac{w}{z}x)\}}\sum_{y\in\mathbb{F}_{q}}\zeta_{p}^{z\operatorname{Tr}_{1}^{e}(y^{\frac{q-1}{2\ell^{k}}}+\delta\frac{w}{z}y)}\right)-1 \nonumber \\
            &=\frac{1}{p^{2}}\left(q^{2}+\epsilon_{f}\sqrt{p^{*}}^{e}(w(0,0)+p-1)+\sum_{z,w\in\mathbb{F}_{p}^{*}}\sigma_{z}\left(\epsilon_{f}\sqrt{p^{*}}^{e}\zeta_{p}^{f^{*}(-\gamma w)}\right)\left(S_{2\ell^{k}}(z,\delta wz)+1\right)\right)-1  \nonumber \\
            &= \frac{1}{p^{2}}\left(q^{2}+\epsilon_{f}\sqrt{p^{*}}^{e}(w(0,0)+p-1)+\epsilon_{f}\sqrt{p^{*}}^{e}\sum_{z\in\mathbb{F}_{p}^{*}}\left(S_{2\ell^{k}}(z,\delta)+1\right)\sum_{w\in\mathbb{F}_{p}^{*}}\zeta_{p}^{w^{2}zf^{*}(\gamma)}\right)-1 \nonumber \\
            &= \frac{1}{p^{2}}\left(q^{2}+\epsilon_{f}\sqrt{p^{*}}^{e}(w(0,0)+p-1)+\epsilon_{f}\sqrt{p^{*}}^{e}\Omega_{2}\right)-1, 
        \end{align}
        where $\Omega_{2}=\sum_{z\in\mathbb{F}_{p}^{*}}\left(S_{2\ell^{k}}(z,\delta)+1\right)\sum_{w\in\mathbb{F}_{p}^{*}}\zeta_{p}^{w^{2}zf^{*}(\gamma)}$. The sixth  equality above holds because for a fixed $z$, as $w$ varies over $\mathbb{F}_{p}^{*}$, $\frac{w}{z}$ also varies.
        \vskip 1pt
        Next, we determine the value of $\Omega_{2}$ by considering the following three cases while $(\gamma,\delta)$ runs through $\mathbb{F}_{q}^{2}\backslash\{(0,0)\}$.
        \vskip 1pt
       $\text{Case }1:$ When $f^{*}(\gamma)=0$, we have
        \begin{align}\label{EQ20}
            \Omega_{2} &=(p-1)\sum_{z\in\mathbb{F}_{p}^{*}}(S_{2\ell^{k}}(z,\delta)+1) \nonumber \\ &= \begin{cases}
                (p-1)w(0,0)+(p-1)^{2};\text{ if }\delta=0\text{ with }\gamma \neq 0, \\
                (p-1)w(0,\delta)+(p-1)^{2};\text{ if }\delta\neq 0.
            \end{cases}
        \end{align}
        $\text{Case }2:$ When $f^{*}(\gamma)\neq 0$ and $\delta=0$, by Lemma $\ref{Le6}$, we have 
        \begin{align}\label{EQ16}
            \Omega_{2} &= \sum_{z\in\mathbb{F}_{p}^{*}}(S_{2\ell^{k}}(z,0)+1)(\eta_{1}(zf^{*}(\gamma))\sqrt{p^{*}}-1) \nonumber \\
             &= \eta_{1}(f^{*}(\gamma))\sqrt{p^{*}}\sum_{z\in\mathbb{F}_{p}^{*}}\eta_{1}(z)S_{2\ell^{k}}(z,0)-w(0,0)-(p-1)
        \end{align}
        From Lemma $\ref{Le2}(a),(b)$ and the fact $\sum_{z\in\mathbb{F}_{p}^{*}}\eta_{1}(z)=0$, one can verify that
        \begin{align}\label{EQ18}
            \sum_{z\in\mathbb{F}_{p}^{*}}\eta_{1}(z)S_{2\ell^{k}}(z,0) &= \frac{q-1}{2\ell^{k}}\sum_{z\in\mathbb{F}_{p}^{*}}\eta_{1}(z)S_{2\ell^{k}}(z) \nonumber \\
            &= \begin{cases}
                \frac{q-1}{2\ell^{k}}(\ell-1)\sqrt{p^{*}}(1+(-1)^{\frac{p-1}{2}});\text{ if }\ell\equiv 1\pmod{p}, \\
                \frac{q-1}{2\ell^{k}}\sqrt{p^{*}}\left(\eta_{1}(t_{1})+(-1)^{\frac{p-1}{2}}\eta_{1}(t_{1})+(\ell-1)(\eta_{1}(t_{2})+(-1)^{\frac{p-1}{2}}\eta_{1}(t_{2}))\right);\text{ if }\\\hspace{8.5cm}\ell\not\equiv 1\pmod{p}, \nonumber  
            \end{cases}      \\
            &= \begin{cases}
                \frac{q-1}{\ell^{k}}(\ell-1)\sqrt{p^{*}};\text{ if }\ell\equiv 1\pmod{p}\text{ and }p\equiv 1\pmod{4},  \\
                 \frac{q-1}{\ell^{k}}\sqrt{p^{*}}(\eta_{1}(t_{1})+(\ell-1)\eta_{1}(t_{2}));\text{ if }\ell\not\equiv 1\pmod{p}\\\hspace{4.4cm}\text{ and }p\equiv 1\pmod{4},  \\
                 0;\text{ if }p\equiv 3\pmod{4},
            \end{cases}
        \end{align}
        where $t_{1}=\phi(\ell^{k})\pmod{p}$ and $t_{2}=\ell^{k-1}\pmod{p}$.
        
        \vskip 1 pt
       $\text{Case }3:$ When $f^{*}(\gamma)\neq 0$ and $\delta\neq 0$, by Lemma $\ref{Le6}$, we have
        \begin{align}\label{EQ19}
            \Omega_{2} &= \sum_{z\in\mathbb{F}_{p}^{*}}(S_{2\ell^{k}}(z,\delta)+1)(\eta_{1}(zf^{*}(\gamma))\sqrt{p^{*}}-1) \nonumber \\
            &= \eta_{1}(f^{*}(\gamma))\sqrt{p^{*}}\sum_{z\in\mathbb{F}_{p}^{*}}\eta_{1}(z)S_{2\ell^{k}}(z,\delta)-w(0,\delta)-(p-1)
        \end{align}
        From Lemma $\ref{Le2}(c),(d)$ and $(\ref{EQ18})$ it is to be noted that
        \begin{align}\label{EQ21}
            \sum_{z\in\mathbb{F}_{p}^{*}}\eta_{1}(z)S_{2\ell^{k}}(z,\delta) &=\sum_{z\in\mathbb{F}_{p}^{*}}\eta_{1}(z)(\chi_{e}(z\xi^{i_{\delta}})\sqrt{q}-\frac{\sqrt{q}+1}{2\ell^{k}}S_{2\ell^{k}}(z)) \nonumber \\
            &= \sum_{z\in\mathbb{F}_{p}^{*}}\eta_{1}(z)\left(\zeta_{p}^{z\operatorname{Tr}_{1}^{e}(\xi^{i_{\delta}})}\sqrt{q}-\frac{\sqrt{q}+1}{2\ell^{k}}S_{2\ell^{k}}(z)\right) \nonumber \\
            &= \begin{cases}
                \sqrt{p^{*}}\left(\eta_{1}(\operatorname{Tr}_{1}^{e}(\xi^{i_{\delta}}))\sqrt{q}-\frac{\sqrt{q}+1}{\ell^{k}}(\ell-1)\right);\text{ if }\ell\equiv 1\pmod{p}\\\hspace{5.4cm}\text{ and }p\equiv 1\pmod{4}, \\
                \sqrt{p^{*}}\left(\eta_{1}(\operatorname{Tr}_{1}^{e}(\xi^{i_{\delta}}))\sqrt{q}-\frac{\sqrt{q}+1}{\ell^{k}}(\eta_{1}(t_{1})+(\ell-1)\eta_{1}(t_{2}))\right);\text{ if }\ell\not\equiv 1\pmod{p}\\\hspace{7.9cm}\text{ and }p\equiv 1\pmod{4}, \\
                \sqrt{p^{*}}\eta_{1}(\operatorname{Tr}_{1}^{e}(\xi^{i_{\delta}}))\sqrt{q};\text{ if }p\equiv 3\pmod{4}.
            \end{cases}
        \end{align}
        By Lemma $\ref{Le3}$, one can easily check 
        \vskip 1pt
        When $\ell\equiv 1\pmod{p}$ and $p\equiv 1\pmod{4}$, we have
        \begin{align*}
            \sum_{z\in\mathbb{F}_{p}^{*}}\eta_{1}(z)S_{2\ell^{k}}(z,\delta) &= \begin{cases}
                \sqrt{p^{*}}\left(\sqrt{q}-\frac{\sqrt{q}+1}{\ell^{k}}(\ell-1)\right);\text{ if }i_{\delta}\in P_{1}^{(3)}\cup P_{3}^{(3)}\cup P_{1}^{(2)}\cup P_{3}^{(2)}, \\
                -\sqrt{p^{*}}\frac{\sqrt{q}+1}{\ell^{k}}(\ell-1);\text{ otherwise, }
            \end{cases}
        \end{align*}
        \vskip 1pt
        when $\ell\not\equiv 1\pmod{p}$ and $p\equiv 1\pmod{4}$, we have
        \begin{align*}
            \sum_{z\in\mathbb{F}_{p}^{*}}\eta_{1}(z)S_{2\ell^{k}}(z,\delta) &= \begin{cases}
                \sqrt{p^{*}}\left(\eta_{1}(t_{1})\sqrt{q}-\frac{\sqrt{q}+1}{\ell^{k}}(\eta_{1}(t_{1})+(\ell-1)\eta_{1}(t_{2}))\right);\text{ if }i_{\delta}\in\{0,\ell^{k}\}, \\
                \sqrt{p^{*}}\left(\eta_{1}(t_{2})\sqrt{q}-\frac{\sqrt{q}+1}{\ell^{k}}(\eta_{1}(t_{1})+(\ell-1)\eta_{1}(t_{2}))\right);\text{ if }i_{\delta}\in P_{1}^{(3)}\cup P_{3}^{(3)}\cup\\\hspace{8.7cm} P_{1}^{(2)}\cup P_{3}^{(2)}, \\
                -\sqrt{p^{*}}\frac{\sqrt{q}+1}{\ell^{k}}(\eta_{1}(t_{1})+(\ell-1)\eta_{1}(t_{2}));\text{ otherwise, }
            \end{cases}
        \end{align*}
        when $\ell\equiv 1\pmod{p}$ and $p\equiv 3\pmod{4}$, we have
        \begin{align*}
            \sum_{z\in\mathbb{F}_{p}^{*}}\eta_{1}(z)S_{2\ell^{k}}(z,\delta) &= \begin{cases}
                -\sqrt{p^{*}q};\text{ if }i_{\delta}\in P_{1}^{(3)}\cup P_{3}^{(3)}, \\
                \sqrt{p^{*}q};\text{ if }i_{\delta}\in P_{1}^{(2)}\cup P_{3}^{(2)}, \\
                0;\text{ otherwise, }
            \end{cases}
        \end{align*}
          and when $\ell\not\equiv 1\pmod{p}$ and $p\equiv 3\pmod{4}$, we have
           \begin{align*}
            \sum_{z\in\mathbb{F}_{p}^{*}}\eta_{1}(z)S_{2\ell^{k}}(z,\delta) &= \begin{cases}
            \sqrt{p^{*}}\eta_{1}(t_{1})\sqrt{q};\text{ if }i_{\delta}=0, \\
                -\sqrt{p^{*}}\eta_{1}(t_{1})\sqrt{q};\text{ if }i_{\delta}=\ell^{k}, \\
                -\sqrt{p^{*}}\eta_{1}(t_{2})\sqrt{q};\text{ if }i_{\delta}\in P_{1}^{(3)}\cup P_{3}^{(3)}, \\
               \sqrt{p^{*}}\eta_{1}(t_{2})\sqrt{q};\text{ if }i_{\delta}\in P_{1}^{(2)}\cup P_{3}^{(2)}, \\
                0;\text{ otherwise, }
            \end{cases}
        \end{align*}
          where $t_{1}=\phi(\ell^{k})\pmod{p}$ and $t_{2}=\ell^{k-1}\pmod{p}$.
          \vskip 1pt
       Then, one can explicitly calculate the values of $\Omega_{2}$ in $(\ref{EQ20})$, $(\ref{EQ16})$ and $(\ref{EQ19})$ by the help of Lemma \ref{Le5}; Lemma \ref{Le5} and $(\ref{EQ18})$; and  Lemma $\ref{Le5}$ and  $(\ref{EQ21})$, respectively.
       \vskip 1pt
        Therefore, $N_{\gamma,\delta}^{(2)}$ is easily determined for each of the above cases by substituting the respective values of  $\Omega_{2}$ in $(\ref{EQ17})$. 
    \end{proof}
\end{lemma}
With the above preparations, we are now ready to determine the parameters and weight distributions of $\mathcal{C}_{D^{'}}$, as stated in the following main
result of this section.
\vspace{1em}
 \begin{theorem}\label{Th2}
 Assume that $q=p^{e}$ and $e=\phi(\ell^{k})$. Let $\mathcal{C}_{D'}=\{(\operatorname{Tr}_{1}^{e}(\gamma x+\delta y))_{(x,y)\in D^{'}}:\gamma,\delta\in\mathbb{F}_{q}\}$, where $D^{'}$ is defined in $(\ref{EQ10})$. Define $T:=\eta_{1}(t_{1})+(\ell-1)\eta_{1}(t_{2})$, where $t_{1}=\phi(\ell^{k})\pmod{p}$ and $t_{2}=\ell^{k-1}\pmod{p}$. Then,
 \begin{enumerate}
 \justifying
     \item[\textnormal{(1)}] If $\ell\equiv 1\pmod{p}$ and $p\equiv 1\pmod{4}$, then $\mathcal{C}_{D^{'}}$ is at most eight-weight $\Large[\frac{q^{2}}{p}+\epsilon_{f}\sqrt{p^{*}}^{e}\left(\frac{q(p-1)}{p}-\frac{(\ell-1)(q-1)}{\ell^{k}}\right)-1,2e\Large]$ linear code over $\mathbb{F}_{p}$ and its weight distribution given in Table $\ref{Table1}$.
 \end{enumerate}
 \begin{table}[htbp]
\centering
\scriptsize
\captionof{table}{}
\label{Table1}
\begin{tabular}{c c} 
\hline
\textnormal{Weight} & \textnormal{Frequency} \\ 
\hline
 \textnormal{0} & \textnormal{1} \\ 
$\frac{q^{2}}{p^{2}}(p-1)$ & $p^{e-1}+\epsilon_{f}(p-1)\frac{\sqrt{p^{*}}^{e}}{p}-1$ \\ 
$\frac{q^{2}}{p^{2}}(p-1)+\epsilon_{f}\sqrt{p^{*}}^{e}\left(\frac{q}{p^{2}}(p-1)^{2}+\frac{\sqrt{q}(p-1)}{p}-\frac{(\ell-1)(q+\sqrt{q})(p-1)}{p\ell^{k}}\right)$ & $\frac{(q-1)(\ell-1)}{\ell^{k}}\left(p^{e-1}+\epsilon_{f}(p-1)\frac{\sqrt{p^{*}}^{e}}{p}\right)$ \\ 
 $\frac{q^{2}}{p^{2}}(p-1)+\epsilon_{f}\sqrt{p^{*}}^{e}\left(\frac{q}{p^{2}}(p-1)^{2}-\frac{(\ell-1)(q+\sqrt{q})(p-1)}{p\ell^{k}}\right)$ & $\frac{(q-1)(\ell^{k}-\ell+1)}{\ell^{k}}\left(p^{e-1}+\epsilon_{f}(p-1)\frac{\sqrt{p^{*}}^{e}}{p}\right)$ \\
 $\frac{q^{2}}{p^{2}}(p-1)+\epsilon_{f}\sqrt{p^{*}}^{e}\left(\frac{q}{p}(p-1)-\frac{(\ell-1)(q-1)}{p\ell^{k}}(p+1)\right)$  & $\frac{(p-1)}{2}\left(p^{e-1}-\epsilon_{f}\frac{\sqrt{p^{*}}^{e}}{p}\right)$ \\
 $\frac{q^{2}}{p^{2}}(p-1)+\epsilon_{f}\sqrt{p^{*}}^{e}\left(\frac{q}{p}(p-1)-\frac{(\ell-1)(q-1)}{p\ell^{k}}(p-1)\right)$  & $\frac{(p-1)}{2}\left(p^{e-1}-\epsilon_{f}\frac{\sqrt{p^{*}}^{e}}{p}\right)$ \\
 $\frac{q^{2}}{p^{2}}(p-1)+\epsilon_{f}\sqrt{p^{*}}^{e}\left(\frac{q}{p^{2}}(p-1)^{2}-\frac{2\sqrt{q}}{p}-\frac{(\ell-1)(\sqrt{q}+1)}{p\ell^{k}}(\sqrt{q}(p-1)-(p+1))\right)$  & $\frac{(q-1)(\ell-1)(p-1)}{2\ell^{k}}\left(p^{e-1}-\epsilon_{f}\frac{\sqrt{p^{*}}^{e}}{p}\right)$ \\
 $\frac{q^{2}}{p^{2}}(p-1)+\epsilon_{f}\sqrt{p^{*}}^{e}\left(\frac{q}{p^{2}}(p-1)^{2}-\frac{(\ell-1)(\sqrt{q}+1)}{p\ell^{k}}(\sqrt{q}(p-1)-(p+1))\right)$  & $\frac{(q-1)(\ell^{k}-\ell+1)(p-1)}{2\ell^{k}}\left(p^{e-1}-\epsilon_{f}\frac{\sqrt{p^{*}}^{e}}{p}\right)$ \\
 $\frac{q^{2}}{p^{2}}(p-1)+\epsilon_{f}\sqrt{p^{*}}^{e}\left(\frac{q}{p^{2}}(p-1)^{2}-\frac{(\ell-1)(q-1)}{p\ell^{k}}(p-1)\right)$  & $\frac{(p-1)(q-1)}{2}\left(p^{e-1}-\epsilon_{f}\frac{\sqrt{p^{*}}^{e}}{p}\right)$ \\
\hline
\end{tabular}
\end{table}
\vspace{1em}
\begin{enumerate}
\justifying
    \item[\textnormal{(2)}] If $\ell\equiv 1\pmod{p}$ and $p\equiv 3\pmod{4}$, then $\mathcal{C}_{D^{'}}$ is at most six-weight $\Large[\frac{q^{2}}{p}+\newline\epsilon_{f}\sqrt{p^{*}}^{e}\left(\frac{q(p-1)}{p}-\frac{(\ell-1)(q-1)}{\ell^{k}}\right)-1,2e\Large]$ linear code over $\mathbb{F}_{p}$ and its weight distribution given in Table $\ref{Table2}$.
\end{enumerate}
\vspace{2mm}
\begin{table}[htbp]
\centering
\scriptsize
\caption{}
\label{Table2}
\begin{tabular}{c c} 
\hline
\textnormal{Weight} & \textnormal{Frequency} \\ 
\hline
 \textnormal{0} & \textnormal{1} \\ 
$\frac{q^{2}}{p^{2}}(p-1)$ & $p^{e-1}+\epsilon_{f}(p-1)\frac{\sqrt{p^{*}}^{e}}{p}-1$ \\ 
$\frac{q^{2}}{p^{2}}(p-1)+\epsilon_{f}\sqrt{p^{*}}^{e}\left(\frac{q}{p^{2}}(p-1)^{2}+\frac{\sqrt{q}(p-1)}{p}-\frac{(\ell-1)(q+\sqrt{q})(p-1)}{p\ell^{k}}\right)$ & $\frac{(q-1)(\ell-1)}{\ell^{k}}\left(p^{e-1}+\epsilon_{f}(p-1)\frac{\sqrt{p^{*}}^{e}}{p}\right)$ \\
$\frac{q^{2}}{p^{2}}(p-1)+\epsilon_{f}\sqrt{p^{*}}^{e}\left(\frac{q}{p^{2}}(p-1)^{2}-\frac{(\ell-1)(q+\sqrt{q})(p-1)}{p\ell^{k}}\right)$ & $\frac{(q-1)(\ell^{k}-\ell+1)}{\ell^{k}}\left(p^{e-1}+\epsilon_{f}(p-1)\frac{\sqrt{p^{*}}^{e}}{p}\right)$ \\

$\frac{q^{2}}{p^{2}}(p-1)+\epsilon_{f}\sqrt{p^{*}}^{e}\left(\frac{q}{p}(p-1)-\frac{(\ell-1)(q-1)}{\ell^{k}}\right)$ & $(p-1)\left(p^{e-1}-\epsilon_{f}\frac{\sqrt{p^{*}}^{e}}{p}\right)$ \\
$\frac{q^{2}}{p^{2}}(p-1)+\epsilon_{f}\sqrt{p^{*}}^{e}\left(\frac{q}{p^{2}}(p-1)^{2}-\frac{2\sqrt{q}}{p}-\frac{(\ell-1)(\sqrt{q}+1)}{p\ell^{k}}(\sqrt{q}(p-1)-p)\right)$ & $\frac{(q-1)(\ell-1)(p-1)}{2\ell^{k}}\left(p^{e-1}-\epsilon_{f}\frac{\sqrt{p^{*}}^{e}}{p}\right)$ \\
$\frac{q^{2}}{p^{2}}(p-1)+\epsilon_{f}\sqrt{p^{*}}^{e}\left(\frac{q}{p^{2}}(p-1)^{2}-\frac{(\ell-1)(\sqrt{q}+1)}{p\ell^{k}}(\sqrt{q}(p-1)-p)\right)$ & $\frac{(q-1)(2\ell^{k}-\ell+1)(p-1)}{2\ell^{k}}\left(p^{e-1}-\epsilon_{f}\frac{\sqrt{p^{*}}^{e}}{p}\right)$ \\
\hline
\end{tabular}
\end{table}
\vspace{1em}
\begin{enumerate}
\justifying
    \item[\textnormal{(3)}] If $\ell\not\equiv 1\pmod{p}$ and $p\equiv 1\pmod{4}$, then $\mathcal{C}_{D^{'}}$ is at most eight-weight \textnormal{(when $\eta_{1}(t_{1})=\eta_{1}(t_{2})$)} and at most nine-weight \textnormal{(when $\eta_{1}(t_{1})=-\eta_{1}(t_{2})$)} $\Large[\frac{q^{2}}{p}+\epsilon_{f}\sqrt{p^{*}}^{e}\left(\frac{q(p-1)}{p}-\frac{(q-1)}{\ell^{k-1}}\right)-1,2e\Large]$ linear code over $\mathbb{F}_{p}$ and their respective weight distributions are given in Table $\ref{Table3}$ and $\ref{Table4}$.
\end{enumerate}
\begin{table}[htbp]
\centering
\scriptsize
\captionof{table}{When $\eta_{1}(t_{1})=\eta_{1}(t_{2})$} 
\label{Table3}
\vspace{0.5em}
\begin{tabular}{c c} 
\hline
\textnormal{Weight} & \textnormal{Frequency} \\ 
\hline
 \textnormal{0} & \textnormal{1} \\ 
$\frac{q^{2}}{p^{2}}(p-1)$ & $p^{e-1}+\epsilon_{f}(p-1)\frac{\sqrt{p^{*}}^{e}}{p}-1$ \\ 
$\frac{q^{2}}{p^{2}}(p-1)+\epsilon_{f}\sqrt{p^{*}}^{e}\left(\frac{q}{p^{2}}(p-1)^{2}+\frac{\sqrt{q}(p-1)}{p}-\frac{(q+\sqrt{q})(p-1)}{p\ell^{k-1}}\right)$ &  $\frac{q-1}{\ell^{k-1}}(p^{e-1}+\epsilon_{f}(p-1)\frac{\sqrt{p^{*}}^{e}}{p})$ \\
$\frac{q^{2}}{p^{2}}(p-1)+\epsilon_{f}\sqrt{p^{*}}^{e}\left(\frac{q}{p^{2}}(p-1)^{2}-\frac{(q+\sqrt{q})(p-1)}{p\ell^{k-1}}\right)$ &  $\frac{(q-1)(\ell^{k-1}-1)}{\ell^{k-1}}(p^{e-1}+\epsilon_{f}(p-1)\frac{\sqrt{p^{*}}^{e}}{p})$ \\
$\frac{q^{2}}{p^{2}}(p-1)+\epsilon_{f}\sqrt{p^{*}}^{e}\left(\frac{q}{p}(p-1)-\frac{(q-1)(p\ell+T)}{p\ell^{k}}\right)$ &  $\frac{(p-1)}{2}(p^{e-1}-\epsilon_{f}\frac{\sqrt{p^{*}}^{e}}{p})$ \\
$\frac{q^{2}}{p^{2}}(p-1)+\epsilon_{f}\sqrt{p^{*}}^{e}\left(\frac{q}{p}(p-1)-\frac{(q-1)(p\ell-T)}{p\ell^{k}}\right)$ &  $\frac{(p-1)}{2}(p^{e-1}-\epsilon_{f}\frac{\sqrt{p^{*}}^{e}}{p})$ \\
$\frac{q^{2}}{p^{2}}(p-1)+\epsilon_{f}\sqrt{p^{*}}^{e}\left(\frac{q}{p^{2}}(p-1)^{2}-\frac{2\sqrt{q}}{p}-\frac{(q-1)}{\ell^{k-1}}+\frac{(\sqrt{q}+1)(\ell\sqrt{q}+\eta_{1}(t_{1})T)}{p\ell^{k}}\right)$ &  $\frac{(q-1)(p-1)}{2\ell^{k-1}}(p^{e-1}-\epsilon_{f}\frac{\sqrt{p^{*}}^{e}}{p})$ \\
$\frac{q^{2}}{p^{2}}(p-1)+\epsilon_{f}\sqrt{p^{*}}^{e}\left(\frac{q}{p^{2}}(p-1)^{2}-\frac{(q-1)}{\ell^{k-1}}+\frac{(\sqrt{q}+1)(\ell\sqrt{q}-\eta_{1}(t_{1})T)}{p\ell^{k}}\right)$ &  $\frac{(p-1)(q-1)}{2}(p^{e-1}-\epsilon_{f}\frac{\sqrt{p^{*}}^{e}}{p})$ \\
$\frac{q^{2}}{p^{2}}(p-1)+\epsilon_{f}\sqrt{p^{*}}^{e}\left(\frac{q}{p^{2}}(p-1)^{2}-\frac{(q-1)}{\ell^{k-1}}+\frac{(\sqrt{q}+1)(\ell\sqrt{q}+\eta_{1}(t_{1})T)}{p\ell^{k}}\right)$ &  $\frac{(p-1)(q-1)(\ell^{k-1}-1)}{2\ell^{k-1}}(p^{e-1}-\epsilon_{f}\frac{\sqrt{p^{*}}^{e}}{p})$ \\
\hline
\end{tabular}
\end{table}

\begin{table}[htbp]
\centering
\scriptsize
\captionof{table}{When $\eta_{1}(t_{1})=-\eta_{1}(t_{2})$} 
\label{Table4}
\vspace{0.5em}
\begin{tabular}{c c} 
\hline
\textnormal{Weight} & \textnormal{Frequency} \\ 
\hline
 \textnormal{0} & \textnormal{1} \\ 
$\frac{q^{2}}{p^{2}}(p-1)$ & $p^{e-1}+\epsilon_{f}(p-1)\frac{\sqrt{p^{*}}^{e}}{p}-1$ \\ 
$\frac{q^{2}}{p^{2}}(p-1)+\epsilon_{f}\sqrt{p^{*}}^{e}\left(\frac{q}{p^{2}}(p-1)^{2}+\frac{\sqrt{q}(p-1)}{p}-\frac{(q+\sqrt{q})(p-1)}{p\ell^{k-1}}\right)$ & $\frac{q-1}{\ell^{k-1}}(p^{e-1}+\epsilon_{f}(p-1)\frac{\sqrt{p^{*}}^{e}}{p})$ \\
$\frac{q^{2}}{p^{2}}(p-1)+\epsilon_{f}\sqrt{p^{*}}^{e}\left(\frac{q}{p^{2}}(p-1)^{2}-\frac{(q+\sqrt{q})(p-1)}{p\ell^{k-1}}\right)$ & $\frac{(q-1)(\ell^{k-1}-1)}{\ell^{k-1}}(p^{e-1}+\epsilon_{f}(p-1)\frac{\sqrt{p^{*}}^{e}}{p})$  \\
$\frac{q^{2}}{p^{2}}(p-1)+\epsilon_{f}\sqrt{p^{*}}^{e}\left(\frac{q}{p}(p-1)-\frac{(q-1)(p\ell+T)}{p\ell^{k}}\right)$ & $\frac{(p-1)}{2}(p^{e-1}-\epsilon_{f}\frac{\sqrt{p^{*}}^{e}}{p})$ \\
$\frac{q^{2}}{p^{2}}(p-1)+\epsilon_{f}\sqrt{p^{*}}^{e}\left(\frac{q}{p}(p-1)-\frac{(q-1)(p\ell-T)}{p\ell^{k}}\right)$ &  $\frac{(p-1)}{2}(p^{e-1}-\epsilon_{f}\frac{\sqrt{p^{*}}^{e}}{p})$  \\
$\frac{q^{2}}{p^{2}}(p-1)+\epsilon_{f}\sqrt{p^{*}}^{e}\left(\frac{q}{p^{2}}(p-1)^{2}-\frac{2\sqrt{q}}{p}-\frac{(q-1)}{\ell^{k-1}}+\frac{(\sqrt{q}+1)(\ell\sqrt{q}+\eta_{1}(t_{1})T)}{p\ell^{k}}\right)$ & $\frac{(q-1)(p-1)}{2\ell^{k}}(p^{e-1}-\epsilon_{f}\frac{\sqrt{p^{*}}^{e}}{p})$ \\
$\frac{q^{2}}{p^{2}}(p-1)+\epsilon_{f}\sqrt{p^{*}}^{e}\left(\frac{q}{p^{2}}(p-1)^{2}-\frac{(q-1)}{\ell^{k-1}}+\frac{(\sqrt{q}+1)(\ell\sqrt{q}+\eta_{1}(t_{1})T)}{p\ell^{k}}\right)$ & $\frac{(q-1)(p-1)(\ell^{k}-1)}{2\ell^{k}}(p^{e-1}-\epsilon_{f}\frac{\sqrt{p^{*}}^{e}}{p})$ \\
$\frac{q^{2}}{p^{2}}(p-1)+\epsilon_{f}\sqrt{p^{*}}^{e}\left(\frac{q}{p^{2}}(p-1)^{2}-\frac{2\sqrt{q}}{p}-\frac{(q-1)}{\ell^{k-1}}+\frac{(\sqrt{q}+1)(\ell\sqrt{q}-\eta_{1}(t_{1})T)}{p\ell^{k}}\right)$ & $\frac{(q-1)(p-1)(\ell-1)}{2\ell^{k}}(p^{e-1}-\epsilon_{f}\frac{\sqrt{p^{*}}^{e}}{p})$ \\
$\frac{q^{2}}{p^{2}}(p-1)+\epsilon_{f}\sqrt{p^{*}}^{e}\left(\frac{q}{p^{2}}(p-1)^{2}-\frac{(q-1)}{\ell^{k-1}}+\frac{(\sqrt{q}+1)(\ell\sqrt{q}-\eta_{1}(t_{1})T)}{p\ell^{k}}\right)$ & $\frac{(q-1)(p-1)(\ell^{k}-\ell+1)}{2\ell^{k}}(p^{e-1}-\epsilon_{f}\frac{\sqrt{p^{*}}^{e}}{p})$ \\
\hline
\end{tabular}
\end{table}
\vspace{1em}
\begin{enumerate}
\justifying
    \item[\textnormal{(4)}] If $\ell\not\equiv 1\pmod{p}$ and $p\equiv 3\pmod{4}$, then $\mathcal{C}_{D'}$ is at most six-weight $\Large[\frac{q^{2}}{p}+\epsilon_{f}\sqrt{p^{*}}^{e}\left(\frac{q(p-1)}{p}-\frac{(q-1)}{\ell^{k-1}}\right)-1,2e\Large]$ linear code over $\mathbb{F}_{p}$ and its weight distribution given in Table $\ref{Table5}$.
\end{enumerate}
\vspace{2mm}
\begin{table}[htbp]
\centering
\scriptsize
\captionof{table}{}
\label{Table5}
\begin{tabular}{c c} 
\hline
\textnormal{Weight} & \textnormal{Frequency} \\ 
\hline
 \textnormal{0} & \textnormal{1} \\ 
$\frac{q^{2}}{p^{2}}(p-1)$ & $p^{e-1}+\epsilon_{f}(p-1)\frac{\sqrt{p^{*}}^{e}}{p}-1$ \\ 
$\frac{q^{2}}{p^{2}}(p-1)+\epsilon_{f}\sqrt{p^{*}}^{e}\left(\frac{q}{p^{2}}(p-1)^{2}+\frac{\sqrt{q}(p-1)}{p}-\frac{(q+\sqrt{q})(p-1)}{p\ell^{k-1}}\right)$ & $\frac{q-1}{\ell^{k-1}}(p^{e-1}+\epsilon_{f}(p-1)\frac{\sqrt{p^{*}}^{e}}{p})$ \\
$\frac{q^{2}}{p^{2}}(p-1)+\epsilon_{f}\sqrt{p^{*}}^{e}\left(\frac{q}{p^{2}}(p-1)^{2}-\frac{(q+\sqrt{q})(p-1)}{p\ell^{k-1}}\right)$ & $\frac{(q-1)(\ell^{k-1}-1)}{\ell^{k-1}}(p^{e-1}+\epsilon_{f}(p-1)\frac{\sqrt{p^{*}}^{e}}{p})$ \\
$\frac{q^{2}}{p^{2}}(p-1)+\epsilon_{f}\sqrt{p^{*}}^{e}\left(\frac{q}{p}(p-1)-\frac{(q-1)}{\ell^{k-1}}\right)$ & $(p-1)(p^{e-1}-\epsilon_{f}\frac{\sqrt{p^{*}}^{e}}{p})$  \\
$\frac{q^{2}}{p^{2}}(p-1)+\epsilon_{f}\sqrt{p^{*}}^{e}\left(\frac{q}{p^{2}}(p-1)^{2}-\frac{2\sqrt{q}}{p}-\frac{(\sqrt{q}+1)}{p\ell^{k-1}}(\sqrt{q}(p-1)-p)\right)$ &  $\frac{(q-1)(p-1)}{2\ell^{k-1}}(p^{e-1}-\epsilon_{f}\frac{\sqrt{p^{*}}^{e}}{p})$  \\
$\frac{q^{2}}{p^{2}}(p-1)+\epsilon_{f}\sqrt{p^{*}}^{e}\left(\frac{q}{p^{2}}(p-1)^{2}-\frac{(\sqrt{q}+1)}{p\ell^{k-1}}(\sqrt{q}(p-1)-p)\right)$ &  $\frac{(q-1)(p-1)(2\ell^{k-1}-1)}{2\ell^{k-1}}(p^{e-1}-\epsilon_{f}\frac{\sqrt{p^{*}}^{e}}{p})$  \\
\hline
\end{tabular}
\end{table}
     \begin{proof}
         We will prove only the case $\ell\equiv 1\pmod{p}$ and $p\equiv 1\pmod{4}$, as the remaining cases can be proved in a similar manner.
         \vskip 1pt
        It is obvious from Lemma \ref{LX1} that when $\ell\equiv 1\pmod{p}$, $\mathcal{C}_{D'}$ has length $n_{2}=|D^{'}|=\frac{q^{2}}{p}+\epsilon_{f}\sqrt{p^{*}}^{e}\left(\frac{q(p-1)}{p}-\frac{(\ell-1)(q-1)}{\ell^{k}}\right)-1$.
        \vskip 1pt
        From $(\ref{EQ15})$, we have the Hamming weight of $c(\gamma,\delta)$ in $\mathcal{C}_{D^{'}}$ is $\operatorname{wt}(c(\gamma,\delta))=n_{2}-N_{\gamma,\delta}^{(2)}$, where $N_{\gamma,\delta}^{(2)}$ is defined in $(\ref{EQ14})$. Next, we determine all possible values of $\operatorname{wt}(c(\gamma,\delta))$ with the help of Lemma \ref{LX2}, while $(\gamma,\delta)$ runs through $\mathbb{F}_{q}^{2}\backslash\{(0,0)\}$. 
        \vskip 1pt
        For $f^{*}(\gamma)=0$ $(\gamma\neq 0)$ and $\delta=0$, we have
        \begin{align*}
            \text{}\operatorname{wt}(c(\gamma,\delta))&=\frac{q^{2}(p-1)}{p^{2}},
            \end{align*}
            For $f^{*}(\gamma)=0$ and $\delta\neq 0$, we have
            \begin{align*}
            \operatorname{wt}(c(\gamma,\delta)) &=\begin{cases}
                \frac{q^{2}(p-1)}{p^{2}}+\epsilon_{f}\sqrt{p^{*}}^{e}\left(\frac{q(p-1)^{2}}{p^{2}}+\frac{\sqrt{q}(p-1)}{p}-\frac{(\ell-1)(q+\sqrt{q})(p-1)}{p\ell^{k}}\right);\text{ if }i_{\delta}\in P_{1}^{(2)}\cup P_{1}^{(3)}\cup \\\hspace{10.1cm}P_{3}^{(2)}\cup P_{3}^{(3)}, \\
                \frac{q^{2}(p-1)}{p^{2}}+\epsilon_{f}\sqrt{p^{*}}^{e}\left(\frac{q(p-1)^{2}}{p^{2}}-\frac{(\ell-1)(q+\sqrt{q})(p-1)}{p\ell^{k}}\right);\text{ otherwise,}
            \end{cases}            
        \end{align*}
        For $f^{*}(\gamma)\neq 0$ and $\delta=0$, we have
        \begin{align*}
            \operatorname{wt}(c(\gamma,\delta)) &=\begin{cases}
                \frac{q^{2}(p-1)}{p^{2}}+\epsilon_{f}\sqrt{p^{*}}^{e}\left(\frac{q(p-1)}{p}-\frac{(\ell-1)(q-1)(p+1)}{p\ell^{k}}\right);\text{ if }\eta_{1}(f^{*}(\gamma))=1, \\
                \frac{q^{2}(p-1)}{p^{2}}+\epsilon_{f}\sqrt{p^{*}}^{e}\left(\frac{q(p-1)}{p}-\frac{(\ell-1)(q-1)(p-1)}{p\ell^{k}}\right);\text{ if }\eta_{1}(f^{*}(\gamma))=-1,
            \end{cases}
        \end{align*}
        For $f^{*}(\gamma)\neq 0$ and $\delta\neq 0$, we have
        \begin{align*}
            \operatorname{wt}(c(\gamma,\delta)) &=\begin{cases}
            \frac{q^{2}(p-1)}{p^{2}}+\epsilon_{f}\sqrt{p^{*}}^{e}\left(\frac{q(p-1)^{2}}{p^{2}}-\frac{(\ell-1)(q-1)(p-1)}{p\ell^{k}}\right);\text{ if }\eta_{1}(f^{*}(\gamma))=-1, \\
            \frac{q^{2}(p-1)}{p^{2}}+\epsilon_{f}\sqrt{p^{*}}^{e}\left(\frac{q(p-1)^{2}}{p^{2}}-\frac{2\sqrt{q}}{p}-\frac{(\ell-1)(\sqrt{q}+1)(\sqrt{q}(p-1)-(p+1))}{p\ell^{k}}\right);\\\hspace{3.7cm}\text{ if }\eta_{1}(f^{*}(\gamma))=1\text{ and }i_{\delta}\in P_{1}^{(2)}\cup P_{1}^{(3)}\cup P_{3}^{(2)}\cup P_{3}^{(3)}, \\
            \frac{q^{2}(p-1)}{p^{2}}+\epsilon_{f}\sqrt{p^{*}}^{e}\left(\frac{q(p-1)^{2}}{p^{2}}-\frac{(\ell-1)(\sqrt{q}+1)(\sqrt{q}(p-1)-(p+1))}{p\ell^{k}}\right);\text{ otherwise. }                
            \end{cases}
        \end{align*}
        Therefore, the eight nonzero weights of $\mathcal{C}_{D^{'}}$ are as follows:
        \begin{align*}
            w_{1}&= \frac{q^{2}(p-1)}{p^{2}}, \\
            w_{2} &=  \frac{q^{2}(p-1)}{p^{2}}+\epsilon_{f}\sqrt{p^{*}}^{e}\left(\frac{q(p-1)^{2}}{p^{2}}+\frac{\sqrt{q}(p-1)}{p}-\frac{(\ell-1)(q+\sqrt{q})(p-1)}{p\ell^{k}}\right), \\
            w_{3} &= \frac{q^{2}(p-1)}{p^{2}}+\epsilon_{f}\sqrt{p^{*}}^{e}\left(\frac{q(p-1)^{2}}{p^{2}}-\frac{(\ell-1)(q+\sqrt{q})(p-1)}{p\ell^{k}}\right), \\
            w_{4} &=  \frac{q^{2}(p-1)}{p^{2}}+\epsilon_{f}\sqrt{p^{*}}^{e}\left(\frac{q(p-1)}{p}-\frac{(\ell-1)(q-1)(p+1)}{p\ell^{k}}\right), \\
            w_{5} &=  \frac{q^{2}(p-1)}{p^{2}}+\epsilon_{f}\sqrt{p^{*}}^{e}\left(\frac{q(p-1)}{p}-\frac{(\ell-1)(q-1)(p-1)}{p\ell^{k}}\right), \\
            w_{6} &= \frac{q^{2}(p-1)}{p^{2}}+\epsilon_{f}\sqrt{p^{*}}^{e}\left(\frac{q(p-1)^{2}}{p^{2}}-\frac{2\sqrt{q}}{p}-\frac{(\ell-1)(\sqrt{q}+1)(\sqrt{q}(p-1)-(p+1))}{p\ell^{k}}\right), \\
            w_{7} &= \frac{q^{2}(p-1)}{p^{2}}+\epsilon_{f}\sqrt{p^{*}}^{e}\left(\frac{q(p-1)^{2}}{p^{2}}-\frac{(\ell-1)(\sqrt{q}+1)(\sqrt{q}(p-1)-(p+1))}{p\ell^{k}}\right), \\
            w_{8} &= \frac{q^{2}(p-1)}{p^{2}}+\epsilon_{f}\sqrt{p^{*}}^{e}\left(\frac{q(p-1)^{2}}{p^{2}}-\frac{(\ell-1)(q-1)(p-1)}{p\ell^{k}}\right).
        \end{align*}
        In the following, we will show their frequencies $A_{w_{1}},A_{w_{2}},A_{w_{3}},A_{w_{4}},A_{w_{5}},A_{w_{6}},A_{w_{7}}$ and $A_{w_{8}}$ respectively. 
        \vskip 1pt
        Note that $e$ is even, then from Lemma \ref{LX3}, we obtain
        \begin{align*}
            A_{w_{1}}&= |\{(\gamma,\delta)\in\mathbb{F}_{q}^{2}: f^{*}(\gamma)=0\text{ with }\gamma\neq 0\text{ and }\delta=0\}|= p^{e-1}+\epsilon_{f}(p-1)\frac{\sqrt{p^{*}}^{e}}{p}-1, \\
            A_{w_{2}} &= |\{(\gamma,\delta)\in\mathbb{F}_{q}^{2}: f^{*}(\gamma)=0\text{ and }\delta\neq 0\text{ with }i_{\delta}\in P_{1}^{(2)}\cup P_{1}^{(3))}\cup P_{3}^{(2)}\cup P_{3}^{(3)}\}| \\
            &= \frac{(q-1)(\ell-1)}{\ell^{k}}\left(p^{e-1}+\epsilon_{f}(p-1)\frac{\sqrt{p^{*}}^{e}}{p}\right), \\
            A_{w_{3}} &= |\{(\gamma,\delta)\in\mathbb{F}_{q}^{2}: f^{*}(\gamma)=0\text{ and }\delta\neq 0\text{ with }i_{\delta}\not\in P_{1}^{(2)}\cup P_{1}^{(3))}\cup P_{3}^{(2)}\cup P_{3}^{(3)}\}| \\
            &= \frac{(q-1)(\ell^{k}-\ell+1)}{\ell^{k}}\left(p^{e-1}+\epsilon_{f}(p-1)\frac{\sqrt{p^{*}}^{e}}{p}\right), \\
            A_{w_{4}} &=|\{(\gamma,\delta)\in\mathbb{F}_{q}^{2}: f^{*}(\gamma)\neq 0\text{ with }\eta_{1}(f^{*}(\gamma))=1\text{ and }\delta=0\}| \\
            &= \frac{(p-1)}{2}\left(p^{e-1}-\epsilon_{f}\frac{\sqrt{p^{*}}^{e}}{p}\right), \\
            A_{w_{5}} &=|\{(\gamma,\delta)\in\mathbb{F}_{q}^{2}: f^{*}(\gamma)\neq 0\text{ with }\eta_{1}(f^{*}(\gamma))=-1\text{ and }\delta=0\}| \\
            &= \frac{(p-1)}{2}\left(p^{e-1}-\epsilon_{f}\frac{\sqrt{p^{*}}^{e}}{p}\right), \\
            A_{w_{6}} &= |\{(\gamma,\delta)\in\mathbb{F}_{q}^{2}: f^{*}(\gamma)\neq 0\text{ such that }\eta_{1}(f^{*}(\gamma))=1\text{ and }\delta\neq 0\text{ with }i_{\delta}\in P_{1}^{(2)}\cup P_{1}^{(3))}\cup P_{3}^{(2)}\cup P_{3}^{(3)}\}| \\
            &= \frac{(q-1)(\ell-1)}{\ell^{k}}\times \frac{(p-1)}{2}\left(p^{e-1}-\epsilon_{f}\frac{\sqrt{p^{*}}^{e}}{p}\right), \\
            A_{w_{7}} &= |\{(\gamma,\delta)\in\mathbb{F}_{q}^{2}: f^{*}(\gamma)\neq 0\text{ such that }\eta_{1}(f^{*}(\gamma))=1\text{ and }\delta\neq 0\text{ with }i_{\delta}\not\in P_{1}^{(2)}\cup P_{1}^{(3))}\cup P_{3}^{(2)}\cup P_{3}^{(3)}\}| \\
            &= \frac{(q-1)(\ell^{k}-\ell+1)}{\ell^{k}}\times \frac{(p-1)}{2}\left(p^{e-1}-\epsilon_{f}\frac{\sqrt{p^{*}}^{e}}{p}\right),  \\
            A_{w_{8}} &= |\{(\gamma,\delta)\in\mathbb{F}_{q}^{2}:\delta\neq 0\text{ and }f^{*}(\gamma)\neq 0\text{ such that }\eta_{1}(f^{*}(\gamma))=-1\}| \\
            &= \frac{(p-1)(q-1)}{2}\left(p^{e-1}-\epsilon_{f}\frac{\sqrt{p^{*}}^{e}}{p}\right).
        \end{align*}
        Obviously, for $\gamma=\delta=0$, the zero codeword occurs only once in $\mathcal{C}_{D'}$, which has exactly $p^{2e}$ codewords, so the dimension of $\mathcal{C}_{D'}$ is $2e$. This completes the proof.
     \end{proof}
 \end{theorem}
 \begin{example}
     Let $f(x)=\operatorname{Tr}_{1}^{e}(x^{2})$ over $\mathbb{F}_{p^{e}}$; then $f(x)$ is a weakly regular bent function with $\epsilon_{f}=-1$ and $l_{f}=2$ from \cite[Table 4]{P12}. Consider $p=3$, $\ell=7$, and $k=1$. Clearly, $\ell\equiv 1\pmod{3}$, $p\equiv 3\pmod{4}$, and $e=\phi(14)=6$. Therefore, $\mathcal{C}_{D'}$ is a ternary $[173420,12,114372]$ linear code, and its weight enumerator is $1+468z^{114372}+27144z^{115182}+146016z^{115344}+162864z^{115668}+194688z^{115830}+260z^{118098}$. Due to the large computation required, using a magma program to verify whether the code $\mathcal{C}_{D'}$ is consistent with Theorem $\ref{Th2}(2)$ is not possible.
 \end{example}
 \vspace{1em}
\begin{example}
    Let $f(x)=\operatorname{Tr}_{1}^{e}(\alpha x^{p^{i}+1})$, where $\alpha$ is a primitive element over $\mathbb{F}_{p^e}$. Then $f(x)$ is a weakly regular bent function with $\epsilon_{f}=1$ and $l_{f}=2$ from \cite[Table 4]{P12}. Consider $p=5$, $\ell=3$, $k=1$, and $i=2$. Clearly, $\ell\not\equiv 1\pmod{5}$, $p\equiv 1\pmod{4}$, $e=\phi(6)=2$, $\eta_{1}(2)=-\eta_{1}(1)=-1$, and $T=\eta_{1}(2)+2\eta_{1}(1)=1$. Therefore, $\mathcal{C}_{D'}$ is a quinary $[104,4,72]$ linear code, and its weight enumerator is $1+8z^{72}+64z^{78}+216z^{80}+128z^{82}+136z^{88}+64z^{92}+8z^{100}$. By a magma program, the obtained code $\mathcal{C}_{D'}$ is consistent with Theorem $\ref{Th2}(3)$ (see Table $\ref{Table4}$).
\end{example}
\vspace{1em}
\begin{example}
    Let $f(x)=\operatorname{Tr}_{1}^{e}(x^{p^{i}+1})$ over $\mathbb{F}_{p^e}$. Then $f(x)$ is a weakly regular bent function with $\epsilon_{f}=-1$ and $l_{f}=2$ from \cite[Table 4]{P12}. Consider $p=5$, $\ell=3$, $k=1$, and $i=2$. Clearly, $\ell\not\equiv 1\pmod{5}$, $p\equiv 1\pmod{4}$, $e=\phi(6)=2$, $\eta_{1}(2)=-\eta_{1}(1)=-1$, and $T=\eta_{1}(2)+2\eta_{1}(1)=1$. Therefore, $\mathcal{C}_{D'}$ is a quinary $[144,4,108]$ linear code, and its weight enumerator is  $1+96z^{108}+204z^{112}+192z^{118}+24z^{120}+96z^{122}+12z^{128}$. By a magma program, the obtained code $\mathcal{C}_{D'}$ is consistent with Theorem $\ref{Th2}(3)$ (see Table $\ref{Table4}$).
\end{example}
\vspace{1em}
\begin{example}
    Let $f(x)=\operatorname{Tr}_{1}^{e}(x^{\frac{3^{i}+1}{2}})$ over $\mathbb{F}_{3^e}$. Then $f(x)$ is a weakly regular bent function with $\epsilon_{f}=-1$ and $l_{f}=2$ from \cite[Table 4]{P12}. Consider $p=3$, $\ell=5$, $k=1$, and $i=5$. Clearly, $\ell\not\equiv 1\pmod{3}$, $p\equiv 3\pmod{4}$, and $e=\phi(10)=4$. Therefore, $\mathcal{C}_{D'}$ is a ternary $[2420,8,1458]$ linear code, and its weight enumerator is $1+20z^{1458}+2400z^{1584}+1680z^{1620}+2400z^{1638}+60z^{1692}$. By a magma program, the obtained code $\mathcal{C}_{D'}$ is consistent with Theorem $\ref{Th2}(4)$.
\end{example}

\section{Concluding remarks}\label{sec5}
Inspired by the work in \cite{P12}, this paper investigates two classes of $p$-ary linear codes through specific defining sets and using the evaluation of the Weil sum $S_{\mathcal{N}}(a,b)$ defined in $(\ref{EQ0})$. 
We presented several classes of linear codes with two-weight, four-weight, six-weight, eight-weight, and nine-weight over $\mathbb{F}_{p}$ and completely determined their parameters and weight distributions (see Theorems \ref{Th1} and \ref{Th2}). Meanwhile, we identify an infinite class of optimal linear codes in our constructions that achieve the Griesmer bound. Notably, some of our newly constructed codes are minimal under certain conditions (see Remark $\ref{Re1}$). 
  

\end{document}